\documentclass[sigconf]{acmart}

\copyrightyear{2022}
\acmYear{2022}
\setcopyright{rightsretained}
\acmConference[CCS '22]{Proceedings of the 2022 ACM SIGSAC Conference on Computer and Communications Security}{November 7--11, 2022}{Los Angeles, CA, USA}
\acmDOI{10.1145/3548606.3560554}
\acmISBN{978-1-4503-9450-5/22/11}

\usepackage{balance}
\usepackage{amsmath,amsfonts,amsthm}
\usepackage[utf8]{inputenc}
\usepackage{booktabs}
\usepackage{graphicx}
\usepackage{cleveref}

\DeclareMathOperator{\exposure}{exposure}
\usepackage{mathtools}
\usepackage{multirow}
\usepackage{caption}
\usepackage[labelfont=normalfont,textfont=normalfont]{subcaption}
\usepackage{enumitem}

\theoremstyle{definition}

\usepackage[ruled,vlined]{algorithm2e}
	\SetKwFor{While}{While}{}{}
	\SetKwFor{For}{For}{}{}
	\SetKwIF{If}{ElseIf}{Else}{If}{}{Else If}{Else}{}
	\SetKw{Return}{Return}

\usepackage{algorithmicx}

\SetCommentSty{mycommfont}

\usepackage[english]{babel}
\usepackage{microtype}
\usepackage{verbatim}
\usepackage{bbm}

\usepackage{tikz}
\newcommand\dottedcircle{\tikz \draw [line cap=round, line width=0.25ex, dash pattern=on 0pt off 2pt] (0,0) circle [radius=0.75ex];}

\usepackage{scalerel}

\newcommand{\added}[1]{#1}

\usepackage{natbib}
\usepackage{graphicx}
\usepackage{xcolor}
\colorlet{DarkRed}{red!70!black}

\newtheorem{game}{Game}[section]

\usepackage{etoolbox}
\makeatletter
\patchcmd{\maketitle}{\@copyrightpermission}{
   \begin{minipage}{0.3\columnwidth}
     \href{https://creativecommons.org/licenses/by/4.0/}{\includegraphics[width=0.90\textwidth]{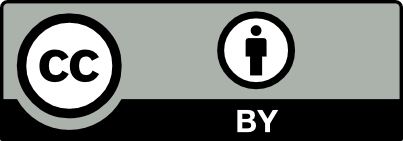}}
   \end{minipage}\hfill
   \begin{minipage}{0.7\columnwidth}
     \href{https://creativecommons.org/licenses/by/4.0/}{This work is licensed under a Creative Commons Attribution International 4.0 License.}
   \end{minipage}

   \vspace{5pt}
}{}{}
\makeatother

\begin{document}

\title{Truth Serum: Poisoning Machine Learning Models\\ to Reveal Their Secrets}

\author{Florian Tram\`er}
\affiliation{ETH Zürich}
\authornote{Authors ordered reverse alphabetically}
\authornote{Work done while the author was at Google}

\author{Reza Shokri}
\affiliation{National University of Singapore}

\author{Ayrton San Joaquin}
\affiliation{Yale-NUS College}

\author{Hoang Le}
\affiliation{Oregon State University}

\author{Matthew Jagielski}
\affiliation{Google}

\author{Sanghyun Hong}
\affiliation{Oregon State University}

\author{Nicholas Carlini}
\affiliation{Google}

\renewcommand{\shortauthors}{Florian Tramèr et al.}
\renewcommand{\shorttitle}{Truth Serum: Poisoning Machine Learning Models to Reveal Their Secrets}

\begin{abstract}
We introduce a new class of attacks on machine learning models.
We show that an adversary who can poison a training dataset can cause models trained on this
dataset to leak significant private details of training points belonging to other parties.
Our active inference attacks connect two independent lines of work targeting the \textit{integrity} and \textit{privacy} of machine learning training data.

Our attacks are effective across membership inference, attribute inference, and data extraction. For example, our targeted attacks can poison \textless$0.1\%$ of the training dataset to boost the performance of inference
attacks by 1 to 2 orders of magnitude.
Further, an adversary who controls a significant fraction of the training data (e.g., 50\%) can launch untargeted attacks that enable $8\times$ more precise inference on \emph{all} other users' otherwise-private data points.

Our results cast doubts on the relevance of cryptographic privacy guarantees in multiparty computation protocols for machine learning, if parties can arbitrarily select their share of training data.
\end{abstract}



\maketitle

\section{Introduction}
A central tenet of computer security is that one cannot obtain any \emph{privacy} without \emph{integrity}~\cite[Chapter 9]{bonehgraduate}.
In cryptography, for example, an adversary who can \emph{modify} a ciphertext, before it is sent to the intended recipient, might be able to leverage this ability to actually \emph{decrypt} the ciphertext.
%
In this paper, we show that this same vulnerability applies to the training of machine learning models.

Currently, there are two long and independent lines of work that study attacks on the integrity and privacy of training data in machine learning (ML). \emph{Data poisoning} attacks~\cite{biggio2012poisoning} target the integrity of an ML model's data collection process to degrade model performance at inference time---either indiscriminately~\cite{biggio2012poisoning,charikar2017learning,jagielski2018manipulating, fowl2021adversarial, munoz2017towards} or on targeted examples~\cite{bhagoji2019analyzing, turner2019label, shafahi2018poison, bagdasaryan2020backdoor, geiping2020witches, liu2017trojaning}. 
Then, separately, privacy attacks such as \emph{membership inference}~\cite{shokri2016membership}, \emph{attribute inference}~\cite{fredrikson2015model, yeom2018privacy} or \emph{data extraction}~\cite{carlini2019secret, carlini2020extracting} aim to infer private information about the model's training set by interacting with a trained model, or by actively participating in the training process~\cite{melis2019exploiting, nasr2019comprehensive}.

\begin{figure}
    \centering
    \begin{subfigure}[b]{0.375\columnwidth}
         \centering
         \includegraphics[height=2.9cm]{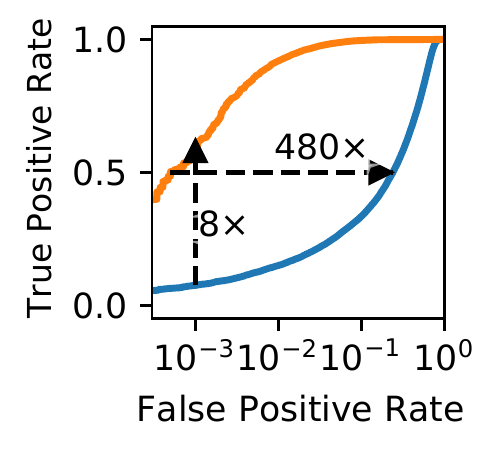}
         \vspace{-1.5em}
         \caption{\footnotesize{Membership Inference}}
         \label{fig:mi}
     \end{subfigure}
     \begin{subfigure}[b]{0.3\columnwidth}
         \centering
         \hspace*{-0.15in}
         \includegraphics[height=3.44cm]{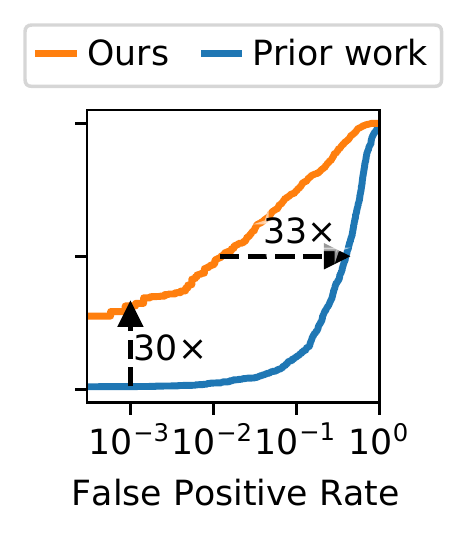}
         \vspace{-1.5em}
         \caption{\footnotesize{Attribute Inference}}
         \label{fig:ai}
     \end{subfigure}
     \begin{subfigure}[b]{0.305\columnwidth}
         \centering
         \hspace*{-0.05in}
         \includegraphics[height=2.9cm]{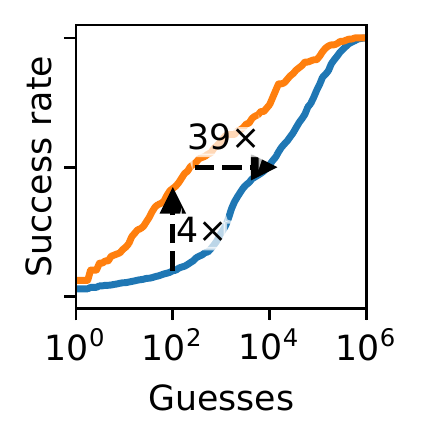}
         \vspace{-1.5em}
         \caption{\footnotesize{Canary Extraction}}
         \label{fig:lm}
     \end{subfigure}
     \vspace{-1em}
     \caption{Poisoning improves an adversary's ability to perform three different privacy attacks.
    (a) For membership inference on CIFAR-10, we improve the true-positive rate (TPR) of~\cite{carlini2021membership} from 7\% to 59\%, at a 0.1\% false-positive rate (FPR). Conversely, at a fixed TPR of 50\%, we reduce the FPR by 480$\times$.
    (b) For attribute inference on Adult (to infer gender), we improve the TPR of~\cite{mehnaz2022your} by 30$\times$.
    (c) To extract 6-digit canaries from WikiText, we reduce the median number of guesses for the attack of~\cite{carlini2019secret} by 39$\times$, from 9{,}018 to 230.}
    \vspace{-0.25em}
\end{figure}

\added{Some works have highlighted connections between these two threats. For example, malicious parties in \emph{federated learning} can craft updates to increase the privacy leakage of other participants~\cite{nasr2019comprehensive, melis2019exploiting, wen2022fishing, hitaj2017deep}. 
Moreover, Chase et al.~\cite{chase2021property} show that poisoning attacks can increase leakage of \emph{global} properties of the training set (e.g., the prevalence of different classes). In this paper, we extend and strengthen these results by demonstrating that an adversary can \emph{statically} poison the training set to maximize the privacy leakage of \emph{individual} training samples belonging to other parties. In other words, we show that the ability to ``\emph{write}'' into the training dataset can be exploited to ``\emph{read}'' from other (private) entries in this dataset.}


We design targeted poisoning attacks on deep learning models that tamper with a small fraction of training data points (\textless$0.1\%$) to improve the performance of membership inference, attribute inference and data extraction attacks on other training examples, by 1 to 2 orders-of-magnitude.
For example, we show that by inserting just $8$ poison samples into the CIFAR-10 training set ($0.03\%$ of the data), an adversary can infer membership of a specific target image with a true-positive-rate (TPR) of $59\%$, compared to $7\%$ without poisoning, at a false-positive rate (FPR) of $0.1\%$. 
Conversely, poisoning enables membership inference attacks to reach 50\% TPR at a FPR of 0.05\%, an error rate $\mathbf{480\times}$ lower than the 24\% FPR from prior work.

Similarly, by poisoning 64 sentences in the WikiText corpus, an adversary can extract a secret 6-digit ``canary''~\cite{carlini2019secret} from a model trained on this corpus with a median of 230 guesses, compared to 9{,}018 guesses without poisoning (an improvement of $\mathbf{39\times}$).

We show that our attacks are robust to uncertainty about the targeted samples, and rigorously investigate the factors that contribute to the success of our attacks.  We find that poisoning has the most impact on samples that originally enjoy the strongest privacy, as our attacks reduce the \emph{average-case} privacy of samples in a dataset to the \emph{worst-case} privacy of data outliers.
We further demonstrate that poisoning drastically lowers the \emph{cost} of state-of-the-art privacy attacks, by alleviating the need for training \emph{shadow models}~\cite{shokri2016membership}.


We then consider \emph{untargeted} attacks where an adversary controls a larger fraction of the training data---as high as 50\%---and aims to increase privacy leakage of \emph{all} other data points. Such attacks are relevant when a small number of parties (e.g., 2) want to jointly train a model on their respective training sets without revealing their own (private) dataset to the other(s), e.g., by using secure multi-party computation~\cite{yao1982protocols, goldreich1987play}.
We show that untargeted poisoning attacks can reduce the error rate of membership inference attacks across all of the victim's data points by a factor of $8\times$.



Our results call into question the relevance of modeling machine learning models as \emph{ideal functionalities} in cryptographic protocols, such as when training models with secure multiparty computation (MPC).
As our attacks show, a malicious party that \emph{honestly} follows the training protocol can exploit their freedom to choose their input data to strongly influence the protocol's ``ideal'' privacy leakage.

%
%
\section{Background and Related Work}
\label{sec:rel_work}

\subsection{Attacks on Training Privacy}
\label{ssec:rel_work_privacy}
Training data privacy is an active research area in machine learning. In our work, we consider three canonical privacy attacks: membership inference~\cite{shokri2016membership}, attribute inference~\cite{fredrikson2015model,fredrikson2014privacy, yeom2018privacy}, and data extraction~\cite{carlini2019secret, carlini2020extracting}. 
In membership inference, an adversary's goal is to determine whether a given sample appeared in the training set of a model or not. 
Participation in a medical trial, for example, may reveal information about a diagnosis~\cite{homer2008resolving}. 
In attribute inference, an adversary uses the model to learn some unknown feature of a given user in the training set. For example, partial knowledge of a user's responses to a survey could allow the adversary to infer the response to other sensitive questions in the survey, by querying a model trained on this (and other) users' responses. Finally, in data extraction, we consider an adversary that seeks to learn a secret string contained in the training data of a language model. We focus on these three canonical attacks as they are the most often considered attacks on training data privacy in the literature.

\subsection{Attacks on Training Integrity}
\label{ssec:rel_work_integrity}
Poisoning attacks can be grouped into three categories: indiscriminate (availability) attacks, targeted attacks, and backdoor (or trojan) attacks. Indiscriminate attacks seek to reduce model performance and render it unusable~\cite{biggio2012poisoning, charikar2017learning, jagielski2018manipulating, fowl2021adversarial, munoz2017towards}. Targeted attacks induce misclassifications for specific benign samples~\cite{shafahi2018poison, suciu2018does, geiping2020witches}. Backdoor attacks add a ``trigger'' into the model, allowing an adversary to induce misclassifications by perturbing arbitrary test points~\cite{bhagoji2019analyzing, turner2019label, bagdasaryan2020backdoor}. Backdoors can also be inserted via supply-chain vulnerabilities, rather than data poisoning attacks~\cite{liu2017neural,liu2017trojaning, gu2019badnets}. However, none of these poisoning attacks have the goal of compromising \emph{privacy}.

Our work considers an attacker that poisons the training data to violate the privacy of other users. Prior work has considered this goal for much stronger adversaries, with additional control over the training procedure.
For example, an adversary that controls part of the training \emph{code} can use the trained model as a side-channel to exfiltrate training data~\cite{song2017machine, bagdasaryan2021blind}.
Or in federated learning, a malicious server can select \emph{model architectures} that enable reconstructing training samples~\cite{boenisch2021curious, fowl2022decepticons}. Alternatively, participants in decentralized learning protocols can boost privacy attacks by sending \emph{dynamic malicious updates}~\cite{nasr2019comprehensive, wen2022fishing, melis2019exploiting, hitaj2017deep}. Our work differs from these in that we only make the weak assumption that the attacker can add a small amount of arbitrary data to the training set \emph{once}, without contributing to any other part of training thereafter.
A similar threat model to ours is considered in~\cite{chase2021property}, for the weaker goal of inferring \emph{global} properties of the training data (e.g., the class prevalences).

\subsection{Defenses}
\label{ssec:rel_work_defenses}
As we consider adversaries that combine poisoning attacks and privacy inference attacks, defenses designed to mitigate either threat may be effective against our attacks.

Defenses against poisoning attacks (either indiscriminate or targeted) design learning algorithms that are robust to some fraction of adversarial data, typically by detecting and removing points that are out-of-distribution~\cite{diakonikolas2019sever, charikar2017learning, jagielski2018manipulating, gupta2019strong, tran2018spectral}.
Defenses against privacy inference either apply heuristics to minimize a model's memorization~\cite{nasr2018machine, jia2019memguard} or train models with differential privacy~\cite{dwork2006calibrating, abadi2016deep}. Training with differential privacy provably protects the privacy of a user's data in \emph{any} dataset, including a poisoned one.

\added{Since our main focus in this work is to introduce a novel threat model that amplifies individual privacy leakage through data poisoning, we design worst-case attacks that are not explicitly aimed at evading specific data poisoning defenses. We note that such poisoning defenses are rarely deployed in practice today. In particular, sanitizing user data in decentralized settings such as federated learning or secure MPC represents a major challenge
~\cite{kairouz2021advances}. In Section~\ref{ssec:mi_ablation_clipping}, we show that a simple loss-clipping approach---inspired by differential privacy---can significantly decrease the effectiveness of our poisoning attacks. Whether our attack techniques can be made robust to such defenses, as well as to more complex data sanitization mechanisms, is an interesting question for future work.}


A related line of work uses poisoning to measure the privacy guarantees of differentially private training algorithms~\cite{jagielski2020auditing, nasr2021adversary}. These works are fundamentally different than ours: they measure the privacy leakage \emph{of the poisoned samples themselves} to investigate worst-case properties of machine learning; in contrast, we show poisoning can harm \emph{other} benign samples.

\subsection{Machine Learning Notation}
\label{ssec:ml_notation}
A classifier $f_\theta: \mathcal{X} \to [0,1]^n$ is  a learned  function  that  maps  an  input  sample $x \in \mathcal{X}$ to a probability vector over $n$ classes. Given a training set $D$ sampled from some distribution $\mathbb{D}$, we let $f_\theta \gets \mathcal{T}(D)$ denote that a classifier with weights $\theta$ is learned by  running  the  training  algorithm $\mathcal{T}$ on  the  training  set $D$.
Given a labeled sample $(x, y)$, we let $\ell(f_\theta(x), y)$ denote a loss function applied to the classifier's output and the ground-truth label, typically the cross-entropy loss.

Causal language models are sequential classifiers that are trained to predict the next word in a sentence. Let sentences in a language be sequences of \emph{tokens} from a set $\mathbb{T}$ (e.g., all English words or sub-words~\cite{wu2016google}).
A generative language model $f_\theta: \mathbb{T}^* \to [0, 1]^{|\mathbb{T}|}$ takes as input a sentence $s$ of an arbitrary number of tokens, and outputs a probability distribution over the value of the next token.
Given a sentence $s=t_1 \dots t_k$ of $k$ tokens, we define the model's loss as:
\begin{equation}
\label{eq:lm_loss}
\ell(f_\theta, s) \coloneqq \frac{1}{k}\sum_{i=0}^{k-1} \ell_{\text{CE}}(f_\theta(t_1 \dots t_i), t_{i+1}), 
\end{equation}
where $\ell_{\text{CE}}$ is the cross-entropy loss and $t_1\dots t_0$ is the empty string.
\section{Amplifying Privacy Leakage \\ with Data Poisoning}
\label{sec:privacy_poisoning}

\textbf{Motivation.}
The fields of security and cryptography are littered with examples where an adversary can turn an attack on integrity into an attack on privacy.
For example, in cryptography a padding oracle attack \cite{bleichenbacher1998chosen, vaudenay2002security} allows an adversary to use their
ability to modify a ciphertext to learn the entire contents of the message.
Similarly, compression leakage attacks~\cite{kelsey2002compression, gluck2013breach} inject data into a user's encrypted traffic (e.g., HTTPS responses) and infer the user's private data by analysing the size of ciphertexts.
Alternatively, in Web security, some past browsers
were vulnerable to attacks wherein the ability to \emph{send} crafted email messages to a victim could be abused to actually \emph{read} the victim's other emails via a Cross-Origin CSS attack \cite{huang2010protecting}.
%
Inspired by these attacks, we show this same type of result is possible in the area of machine learning.

\subsection{Threat Model}
\label{ssec:threat_model}

We consider an adversary $\mathcal{A}$ that can inject some data $D_{\text{adv}}$ into a machine learning model's training set $D$. The goal of this adversary is to amplify their ability to infer information about the contents of $D$, by interacting with a model trained on $D \cup D_{\text{adv}}$.
In contrast to prior attacks on distributed or federated learning~\cite{nasr2019comprehensive, melis2019exploiting}, our adversary cannot actively participate in the learning process. The adversary can only statically poison their data once, and after this can only interact with the final trained model. 

\paragraph{The privacy game.}

We consider a generic privacy game, wherein the adversary has to guess which element from some \emph{universe}~$\mathcal{U}$ was used to train a model. By appropriately defining the universe~$\mathcal{U}$ this game generalizes a number of prior privacy attack games, from membership inference to data extraction.

\begin{game}[Privacy Inference Game]
    \label{game:privacy}
    The game proceeds between a challenger $\mathcal{C}$ and an adversary $\mathcal{A}$. Both have  access to a distribution~$\mathbb{D}$, and know the universe $\mathcal{U}$ and training algorithm $\mathcal{T}$.
    \begin{enumerate}
        \item The challenger samples a dataset $D \gets \mathbb{D}$ and a target $z \gets \mathcal{U}$ from the universe (such that $D \cap \mathcal{U} = \emptyset$).
        \item The challenger trains a model $f_\theta \gets \mathcal{T}(D \cup \{z\})$ on the dataset $D$ and target $z$. 
        \item The challenger gives the adversary query access to $f_\theta$.
        \item The adversary emits a guess $\hat{z} \in \mathcal{U}$.
        \item The adversary wins the game if $\hat{z} = z$.
    \end{enumerate}
\end{game}

The universe $\mathcal{U}$ captures the adversary's \emph{prior} belief about the possible value that the targeted example may take.  
In the membership inference game (see~\cite{yeom2018privacy,  jayaraman2020revisiting}), for a specific target example $x$ the universe is $\mathcal{U} = \{x, \bot\}$---where $\bot$ indicates the absence of an example. That is, the adversary guesses whether the model $f$ is trained on $D$ or on $D \cup \{x\}$. For attribute inference, the universe~$\mathcal{U}$ contains the real targeted example $x$, along with all ``alternate versions'' of $x$ with other values for an unknown attribute of $x$.
Attacks that extract well-formatted sensitive values, such as credit card numbers~\cite{carlini2019secret}, can be modeled with a universe $\mathcal{U}$ of all possible values that the secret could take.

We now introduce our new privacy game, which adds the ability for an adversary to poison the dataset.
This is a strictly more general game, with the objective of maximizing the privacy leakage of the targeted point.
The changes to Game~\ref{game:privacy} are highlighted in red.

\begin{game}[Privacy Inference Game with Poisoning]
    \label{game:privacy_poison}
    The game proceeds between a challenger $\mathcal{C}$ and an adversary $\mathcal{A}$. Both have  access to a distribution $\mathbb{D}$, and know the universe $\mathcal{U}$ and training algorithm $\mathcal{T}$.
    \begin{enumerate}
        \item The challenger samples a dataset $D \gets \mathbb{D}$ and a target $z \gets \mathcal{U}$ from the universe (such that $D \cap \mathcal{U} = \emptyset$).
        \item {\color{DarkRed} The adversary sends a poisoned dataset $D_{\text{adv}}$ of size $N_{\text{adv}}$ to the challenger.}
        \item The challenger trains a model $f_\theta \gets \mathcal{T}(D \ {\color{DarkRed}\cup\ D_{\text{adv}}} \cup \{z\})$ on the {\color{DarkRed}poisoned dataset $D \cup D_{\text{adv}}$} and target $z$. 
        \item The challenger gives the adversary query access to $f_\theta$.
        \item The adversary emits a guess $\hat{z} \in \mathcal{U}$.
        \item The adversary wins the game if $\hat{z} = z$.
    \end{enumerate}
\end{game}


\subsubsection{Adversary Capabilities}

The above poisoning game implicitly assumes a number of adversarial capabilities, which we now discuss more explicitly.

\Cref{game:privacy_poison} assumes that the adversary knows the data distribution~$\mathbb{D}$ and the universe of possible target values $\mathcal{U}$. These capabilities are standard and easy to meet in practice.
The adversary further gets to add a set of $N_{\text{adv}}$ poisoned points into the training set. We will consider attacks that require adding only a small number of targeted poisoned points (as low as $N_{\text{adv}}=1$), as well as attacks that assume much larger data contributions (up to $N_{\text{adv}} = |D|$) as one could expect in MPC settings with a small number of parties.

We impose no restrictions on the adversary's poisons being ``stealthy''. That is, we allow for the poisoned dataset $D_{\text{adv}}$ to be arbitrary. 
As we will see, designing poisoning attacks that maximize privacy leakage is non-trivial---even when the adversary is not constrained in their choice of poisons.
As poisoning attacks that target data privacy have not been studied so far, we aim here to understand how effective such attacks could be in the worst case, and leave the study of attacks with further constraints (such as ``clean label'' poisoning~\cite{turner2019label, shafahi2018poison}) to future work.

Finally, the game assumes that the adversary targets a specific example $z$. We call this a \emph{targeted attack}. 
We also consider \emph{untargeted} attacks in \Cref{ssec:mi_untargeted}, where the attacker crafts a poisoned dataset $D_{\text{adv}}$ to harm the privacy of \emph{all} samples in the training set $D$. 

\subsubsection{Success Metrics}

When the universe of secret values is small (as for membership inference, where $|\mathcal{U}|=2$, or for attribute inference where it is the cardinality of the attribute), we measure an attack's success rate by its true-positive rate (TPR) and false-positive rate (FPR) over multiple iterations of the game. Following~\cite{carlini2021membership}, we focus in particular on the attack performance at low false-positive rates (e.g., FPR=$0.1\%$), which measures the attack's propensity to precisely target the privacy of some worst-case users.

For \textbf{membership inference}, we naturally define a true-positive as a correct guess of membership, i.e., $\hat{z}=z$ when $z=x$, and a false-positive as an incorrect membership guess, $\hat{z}\neq z$ when $z=x$.

For \textbf{attribute inference}, we define a ``positive'' as an example with a specific value for the unknown attribute (e.g., if the unknown attribute is gender, we define ``female'' as the positive class).

For \textbf{canary extraction}, where the universe of possible target values is large (e.g., all possible credit card numbers), we amend \Cref{game:privacy_poison} to allow the adversary to obtain ``partial credit'' by emitting multiple guesses. Specifically, following~\cite{carlini2019secret}, we let the adversary output an \emph{ordering} (a permutation) $\hat{Z} = \pi(\mathcal{U})$ of the secret's possible values, from most likely to least likely. We then measure the attack's success by the \emph{exposure}~\cite{carlini2019secret} (in bits) of the correct secret $z$:
\begin{equation}
\label{eq:exposure}
    \exposure(z; \hat{Z}) \coloneqq \log_2\left(|\mathcal{U}|\right) - \log_2\left( \textbf{rank}(z; \hat{Z}) \right) \;.
\end{equation}
The exposure ranges from $0$ bits (when the correct secret $z$ is ranked as the least likely value), to $\log_2(|\mathcal{U}|)$ bits (when the adversary's most likely guess is the correct value $z$).

\subsection{Attack Overview}
\label{ssec:attack}

We begin with a high-level overview of our poisoning attack strategies. For simplicity of exposition, we focus on the special case of membership inference. Our attacks for attribute inference and canary extraction follow similar principles.

Given a target sample $(x, y)$, the standard privacy game (for membership inference) in \Cref{game:privacy} asks the adversary to distinguish two worlds, where the model is respectively trained on $D \cup \{(x, y)\}$ or on $D$.
When we give the adversary the ability to poison the dataset in \Cref{game:privacy_poison}, the goal is now to alter the dataset $D$ so that the above two worlds become easier to distinguish.

Note that this goal is very different from simply maximizing the model's memorization of the target $(x, y)$. This could be achieved with the following (bad) strategy: poison the dataset $D$ by adding multiple identical copies of $(x, y)$ into it. This will ensure that the trained model $f_\theta$ strongly memorizes the target (i.e., the model will correctly classify $x$ with very high confidence). However, this will be true in both worlds, regardless of whether the target $(x, y)$ was in the original training set $D$ or not. This strategy thus does not help the adversary in solving the distinguishing game---and in fact actually makes it \emph{more difficult} to distinguish membership.

Instead, the adversary should alter the training set $D$ so as to maximize the \emph{influence} of the target $(x, y)$. That is, we want the poisoned training set $D \cup D_{\text{adv}}$ to be such that the inclusion of the target $(x, y)$ provides a maximal \emph{change} in the trained model's behavior on some inputs of the adversary's choice.

To illustrate this principle, we begin by demonstrating a provably \emph{perfect} privacy-poisoning attack for the special case of nearest-neighbor classifiers. We also propose an alternative attack for SVMs in \Cref{apx:theory}.
We then describe our design principles for empirical attacks on deep neural networks.

\paragraph{\textbf{Warm-up:} provably amplifying membership leakage in kNNs}

Consider a $k$-Nearest Neighbor (kNN) classifier (assume, \emph{wlog.,} that $k$ is odd). Given a labeled training set $D$, and a test sample $x$, this classifier finds the $k$ nearest neighbors of $x$ in $D$, and outputs the majority label among these $k$ neighbors.
We assume the attacker has black-box query access to the trained classifier.

We demonstrate how to poison a kNN classifier so that the classifier labels a target example $(x,y)$ correctly \emph{if and only if} the target is in the original training set $D$. This attack thus lets the adversary win the membership inference game with $100\%$ accuracy.

Our poisoning attack (see Algorithm~\ref{alg:knnpois} in \Cref{apx:theory}) creates a dataset $D_{\text{adv}}$ of size $k$ that contains $k-1$ copies of the target $x$, half correctly labeled as $y$ and half mislabeled as $y' \neq y$. We further add one poisoned example $x'$ at a small distance $\delta$ from $x$ and also mislabeled as $y'$ (we assume that no other point in the training set $D$ is within distance $\delta$ from $x$).
This attack maximizes the \emph{influence} of the targeted point, by turning it into a tie-breaker for classifying $x$ when it is a member.

The attacker infers that the target example $(x,y)$ is a member, if and only if the trained model correctly classifies $x$ as class $y$. To see that the attack works, consider the two possible worlds:
\begin{itemize}
    \item \emph{The target is in $D$}: There are $k$ copies of $x$ in the poisoned training set $D \cup D_{\text{adv}}$: the $k-1$ poisoned copies (half are correctly labeled) and the target $(x, y)$. Thus, the majority vote among the $k$ neighbors yields the correct class $y$.
    \item \emph{The target is not in $D$}: As all points in $D$ are at distance at least $\delta$ from the target $x$, the $k$ neighbors selected by the model are the adversary's $k$ poisoned points, a majority of which are mislabeled as $y'$. Thus, the model outputs $y'$.
\end{itemize}

In \Cref{apx:theory}, we show that our attack is non-trivial, in that there exist points for which poisoning is \emph{necessary} to achieve perfect membership inference. In fact, we show that for some points, a non-poisoning adversary cannot infer membership better than chance.

\paragraph{Amplifying privacy leakage in deep neural networks}

The above attack on kNNs exploits the classifier's specific structure which lets us turn any example's membership into a perfect tie-breaker for the model's decision on that example.
In deep neural networks, it is unlikely that examples can exhibit such a clear cut influence (i.e., due to the stochasticity of training, it is unlikely that a specific model behavior would occur \emph{if and only if} an example is a member).

Instead, we could try to cast the adversary's goal as an \emph{optimization problem}, of selecting a poisoned dataset $D_{\text{adv}}$ that maximizes the distinguishability of models trained with or without the target $(x,y)$. Yet, solving such an optimization problem is daunting. While prior work does optimize poisons to maximally alter a \emph{single} model's confidence on a specific target point~\cite{shafahi2018poison, turner2019label, zhu2019transferable}, here we would instead need to optimize for a \emph{difference in distributions} of the decisions of two models trained on two neighboring datasets.

Rather than tackle this optimization problem directly, we ``handcraft'' strategies that empirically increase a sample's influence on the model.
We start from the observation in prior work that the most vulnerable examples to privacy attacks are data \emph{outliers}~\cite{yeom2018privacy, carlini2021membership}. Such examples are easy to attack precisely because they have a large influence on the model: a model trained on an outlier has a much lower loss on this sample than a model that was not trained on it.
Yet, in our threat model, the attacker cannot control or modify the targeted example $x$ (and $x$ is unlikely, \emph{a priori}, to be an outlier). Our insight then is to poison the training dataset so as to \emph{transform} the targeted example $x$ into an outlier. For example, we could fool the model into believing that the targeted point $x$ is \emph{mislabeled}. Then, the presence of the correctly labeled target $(x, y)$ in the training set is likely to have a large influence on the model's decision.

In \Cref{sec:mi}, we show how to instantiate this attack strategy to boost membership inference attacks on standard image datasets. We then extend this attack strategy in \Cref{sec:ai} to the case of attribute inference attacks for tabular datasets. Finally, in \Cref{sec:lm} we propose attack strategies tailored to language models, that maximize the leakage of specially formatted canary sequences.

\section{Membership Inference Attacks}
\label{sec:mi}

Membership inference (MI) captures one of the most generic notions of privacy leakage in machine learning. Indeed, any form of data leakage from a model's training set (e.g., attribute inference or data extraction) implies the ability to infer membership of some training examples.
As a result, membership inference is a natural target for evaluating the impact of poisoning attacks on data privacy.

In this section, we introduce and analyze data poisoning attacks that improve membership inference by one to two orders of magnitude.
\Cref{ssec:mi_targeted} describes a \emph{targeted} attack that increases leakage of a specific sample $(x, y)$, and \Cref{ssec:mi_analysis} contains an analysis of this attack's success. \Cref{ssec:mi_untargeted} explores \emph{untargeted} attacks that increase privacy leakage on \emph{all} training points simultaneously.

\subsection{Experimental Setup}
\label{ssec:mi_setup}

We extend the recent attack of \cite{carlini2021membership} that performs membership inference via a per-example log-likelihood test.
The attack first trains $N$ \emph{shadow models} such that each sample $(x,y)$ appears in the training set of half of the shadow models, and not in the other half.
We then compute the losses of both sets of models on $x$:
\begin{align*}
{L}_{\text{in}} &= \{ \ell(f(x), y) \ :\ f \text{ trained on } (x,y)\} \ , \\
{L}_{\text{out}} &= \{ \ell(f(x), y) \ :\ f \text{ not trained on } (x,y)\} 
\end{align*}
and fit Gaussian distributions $\mathcal{N}(\mu_{\text{in}}, \sigma_{\text{in}}^2)$ to $L_{\text{in}}$, and  $\mathcal{N}(\mu_{\text{out}}, \sigma_{\text{out}}^2)$ to $L_{\text{out}}$ (with a logit scaling of the losses, as in \cite{carlini2021membership}).
Then, to infer membership of $x$ in a trained model $f_\theta$, we compute the loss of $f_\theta$ on $x$, and perform a standard likelihood-ratio test for the hypotheses that $x$ was drawn from $\mathcal{N}(\mu_{\text{in}}, \sigma_{\text{in}}^2)$ or from $\mathcal{N}(\mu_{\text{out}}, \sigma_{\text{out}}^2)$.

To amplify the attack with poisoning, the adversary builds a poisoned dataset $D_{\text{adv}}$ that is added to the training set of $f_\theta$. The adversary also adds $D_{\text{adv}}$ to each shadow model's training set (so that these models are as similar as possible to the target model $f_\theta$).

We perform our experiments on CIFAR-10 and CIFAR-100~\cite{cifar}---standard image datasets of 50{,}000 samples from respectively 10 and 100 classes. The target models (and shadow models) use a Wide-ResNet architecture~\cite{zagoruyko2016wide} trained for 100 epochs with weight decay and common data augmentations (random image flips and crops). For each dataset, we train $N=128$ models on random $50\%$ splits of the original training set.\footnote{The training sets of the target model and shadow models thus partially overlap (although the adversary does not know which points are in the target's training set). Carlini et al.~\cite{carlini2021membership} show that their attack is minimally affected if the attacker's shadow models are trained on datasets fully disjoint from the target's training set.} The models achieve 91\% test accuracy on CIFAR-10 and 67\% test-accuracy on CIFAR-100 on average.

\begin{figure}[t]
    \centering
    \includegraphics[width=0.8\columnwidth]{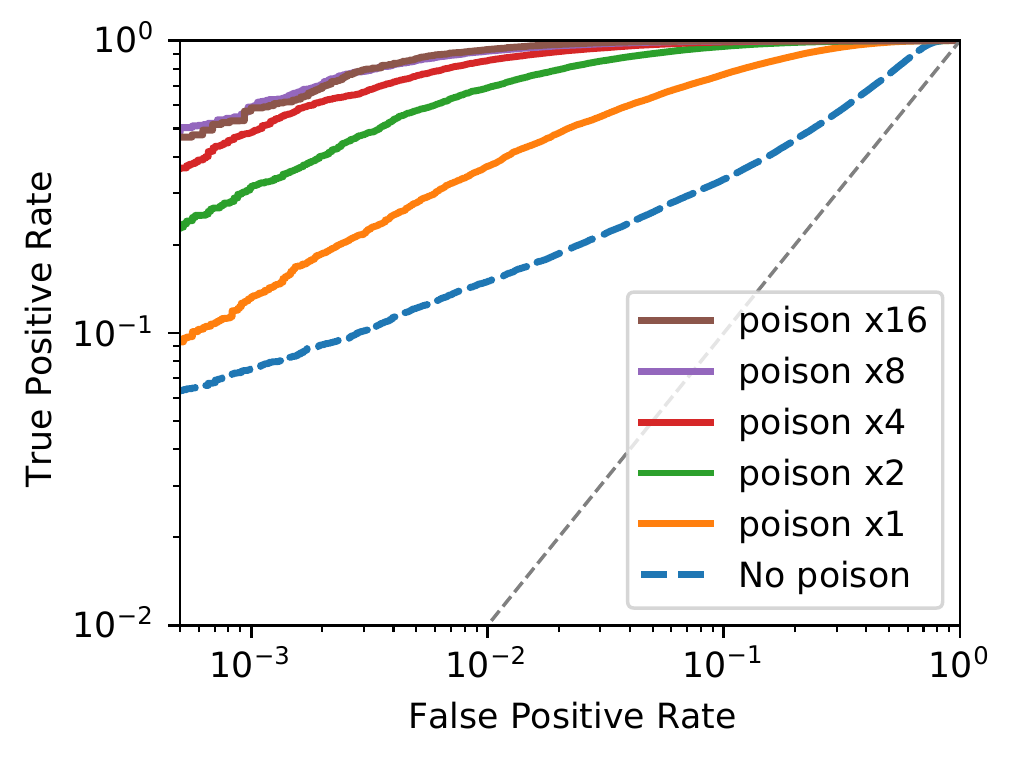}
    \caption{Targeted poisoning attacks boost membership inference on CIFAR-10. For 250 random data points, we insert $1$ to $16$ mislabelled copies of the point into the training set, and run the MI attack of~\cite{carlini2021membership} with 128 shadow models.}
    \label{fig:mi_cifar10}
\end{figure}

\subsection{Targeted Poisoning Attacks}
\label{ssec:mi_targeted}

We now design our poisoning attack to increase the membership inference success rate for a specific target example $x$. That is, the attacker knows the data of $x$ (but not whether it is used to train the model) and designs a poisoned dataset $D_{\text{adv}}$ \emph{adaptively} based on $x$. 


\paragraph{Label flipping attacks.}
We find that \emph{label flipping} attacks are a very powerful form of poisoning attacks to increase data leakage. Given a targeted example $x$ with label $y$, the adversary inserts the mislabelled poisons $D_{\text{adv}} = \{(x, y'), \dots, (x, y')\}$ for some label $y' \neq y$. The rationale for this attack is that a model trained on $D_{\text{adv}}$ will learn to associate $x$ with label $y'$, and the now ``mislabelled'' target $(x, y)$ will be treated as an outlier and have a heightened influence on the model when present in the training set.

To instantiate this attack on CIFAR-10 and CIFAR-100, we pick $250$ targeted points at random from the original training set. For each targeted example $(x, y)$, the poisoned dataset $D_{\text{adv}}$ contains a mislabelled example $(x, y')$ replicated $r$ times, for $r \in \{1, 2, 4, 8, 16\}$.
We report the average attack performance for a full leave-one-out cross-validation (i.e., we evaluate the attack 128 times, using one model as the target and the rest as shadow models).

\paragraph{Results.}
\Cref{fig:mi_cifar10} and \Cref{fig:mi_cifar100} (appendix) show the performance of our membership inference attack on CIFAR-10 and CIFAR-100 respectively, as we vary the number of poisons $r$ per sample.

We find that this attack is remarkably effective. Even with a single poisoned example ($r=1$), the attack's true-positive rate (TPR) at a $0.1\%$ false-positive rate (FPR) increases by $1.75\times$. With $8$ poisons ($0.03\%$ of the model's training set size), the TPR increases by a factor $8\times$ on CIFAR-10, from $7\%$ to $59\%$.
On CIFAR-100, poisoning increases the baseline's strong TPR of $22\%$ to $69\%$ at a FPR of $0.1\%$.

Alternatively, we could aim for a fixed recall and use poisoning to reduce the MI attack's error rate. Without poisoning, an attack that correctly identifies half of the targeted CIFAR-10 members (i.e., a TPR of $50\%$) would also incorrectly label $24\%$ of non-members as members. With poisoning, the same recall is achieved while only mislabeling $0.05\%$ of non-members---\textbf{a factor $\mathbf{480\times}$ improvement}. On CIFAR-100, also for a $50\%$ TPR, poisoning reduces the attack's false-positive rate \textbf{by a factor $\mathbf{100\times}$}, from $2.5\%$ to $0.025\%$.

\begin{figure}[t]
    \centering
    \includegraphics[width=0.8\columnwidth]{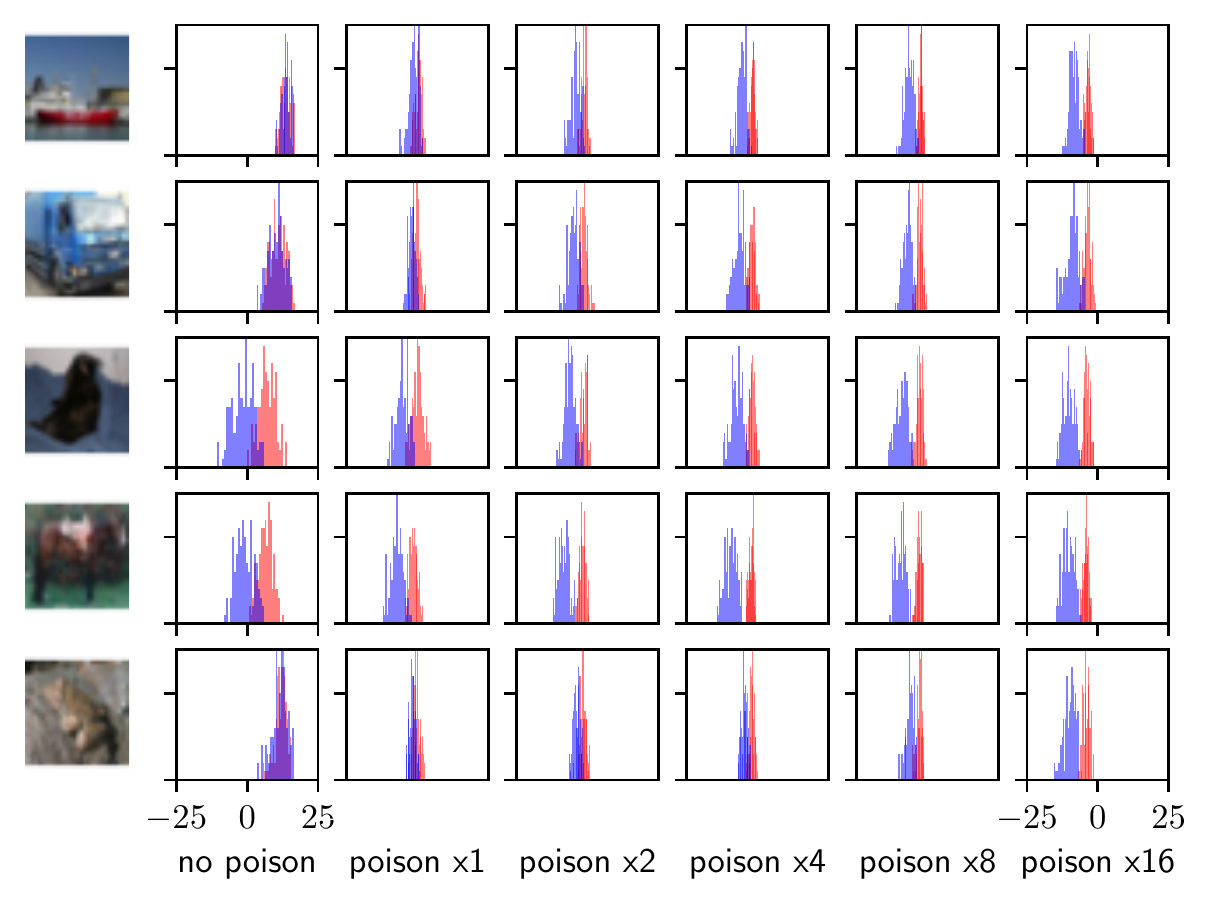}
    \caption{Our poisoning attack separates the loss distributions of members and non-members, making them more distinguishable. For five random CIFAR-10 examples, we plot the (logit-scaled) loss distribution on that example when it is a member (red) or not (blue). The horizontal axis varies the number of times the adversary poisons the example.}
    \label{fig:mi_cifar10_dists_CE}
\end{figure}

As we run multiple targeted attacks simultaneously (for efficiency sake), the total number of poisons is large (up to $4{,}000$ mislabelled points). Yet, the poisoned model's test accuracy is minimally reduced (from $92\%$ to $88\%$) and the MI success rate on non-targeted points remains unchanged. Thus, we are not compounding the effects of the $250$ targeted attacks. As a sanity check, we repeat the experiment with only $50$ targeted points, and obtain similar results.

\subsection{Analysis and Ablations}
\label{ssec:mi_analysis}

We have shown that targeted poisoning attacks significantly increase membership leakage. We now set out to understand the principles underlying our attack's success.

\subsubsection{Why does our attack work?}

In \Cref{fig:mi_cifar10_dists_CE} we plot the distribution of model confidences for five CIFAR-10 examples, when the example is a member (in red) and when it is not (in blue). 
On the horizontal axis, we vary the number of poisons (i.e., how many times this example is mislabeled in the training set).
Without poisoning (left column), the distributions overlap significantly for most examples.
%
As we increase the number of poisons, the confidences shift significantly to the left, as the model becomes less and less confident in the example's true label. But crucially, the distributions also become easier to separate, because the (relative) influence of the targeted example on the trained model is now much larger.

To illustrate, consider the top example in \Cref{fig:mi_cifar10_dists_CE} (labeled ``ship''). Without poisoning, this example's confidence is in the range [99.99\%, 100\%] when it is a member, and [99.98\%, 100\%] when it is not. Confidently inferring membership is thus impossible. With 16 poisons, however, the confidence on this example is in the range [0.4\%, 28.5\%] when it is a member, and [0\%, 2.4\%] when it is not---thus enabling precise membership inference when the confidence exceeds $2.4\%$.

\begin{figure}[t]
    \centering
    \includegraphics[width=0.8\columnwidth]{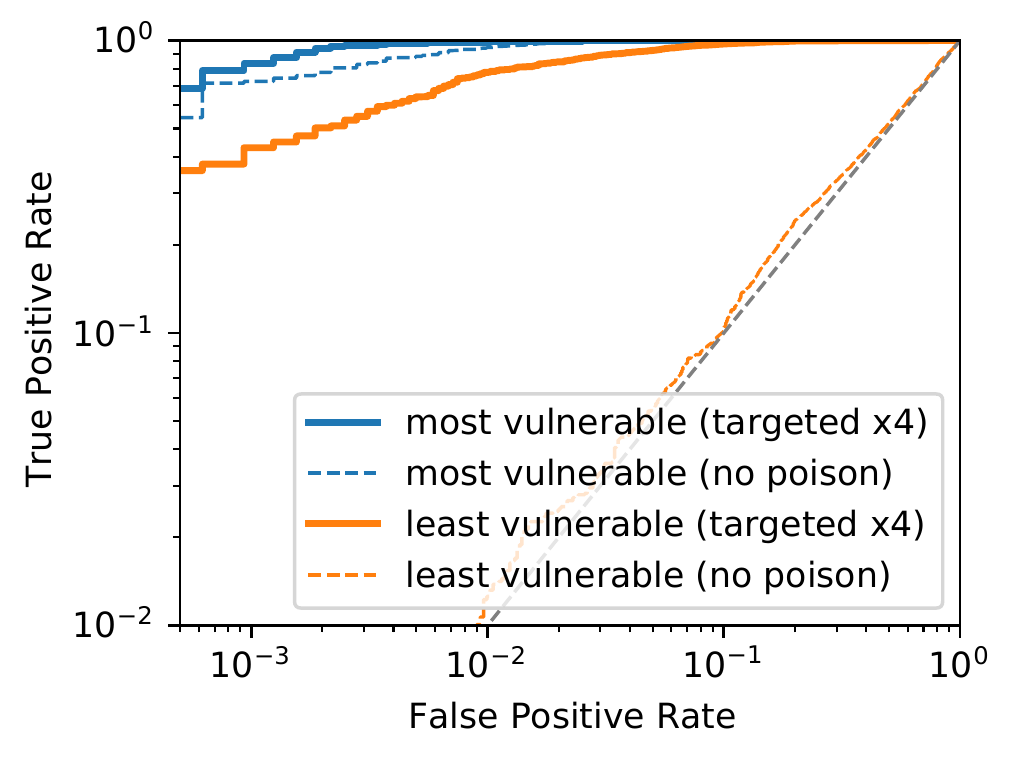}
    \caption{Poisoning causes previously-safe data points to become vulnerable.
    We run our attack for the 5\% of points that are originally most- and least-vulnerable to membership inference without poisoning. While poisoning has little effect for the most vulnerable points, poisoning the least vulnerable points improves the TPR at a 0.1\% FPR by a factor $430\times$.}
    \label{fig:mi_extreme}
\end{figure}

\subsubsection{Which points are vulnerable to our attack?}
\label{sssec:extreme_points}
Our poisoning attack could increase the MI success rate in different ways. Poisoning could increase the attack accuracy uniformly across all data points, or it might disparately impact some data points. We show that the latter is true: \textbf{our attack disparately impacts \emph{inliers} that were originally \emph{safe} from membership inference}.
This result has striking consequences: even if a user is an inlier and therefore might not be worried about privacy leakage, an active poisoning attacker that targets this user can still infer membership.

In \Cref{fig:mi_extreme}, we show the performance of our poisoning attack on those data points that are initially easiest and hardest to infer membership for.
We run the membership inference attack of~\cite{carlini2021membership} on all CIFAR-10 points, and select the 5\% of samples where the attack succeeds least often and most often (averaged over all 128 models).
We then re-run the baseline attack on these extremal points with a new set of models (to ensure our selection of points did not overfit) and compare with our label flipping attack with $r=4$.

Poisoning has a minor effect on data points that are already outliers: here even the baseline MI attack has a high success rate (73\% TPR at a 0.1\% FPR) and thus there is little room for improvement.\footnote{In \Cref{fig:mi_cifar10_dists_CE}, we see that examples for which MI succeeds \emph{without} poisoning tend to already be outliers. For example, the third and fourth example from the top are a ``bird'' mislabelled as ``cat'' in the CIFAR-10 training set, and a ``horse'' confused as a ``deer''.} For points that are originally hardest to attack, however, poisoning improves the attack's TPR \textbf{by a factor $\mathbf{430\times}$}, from 0.1\% to 43\%.

\begin{figure}[t]
    \centering
    \includegraphics[width=0.8\columnwidth]{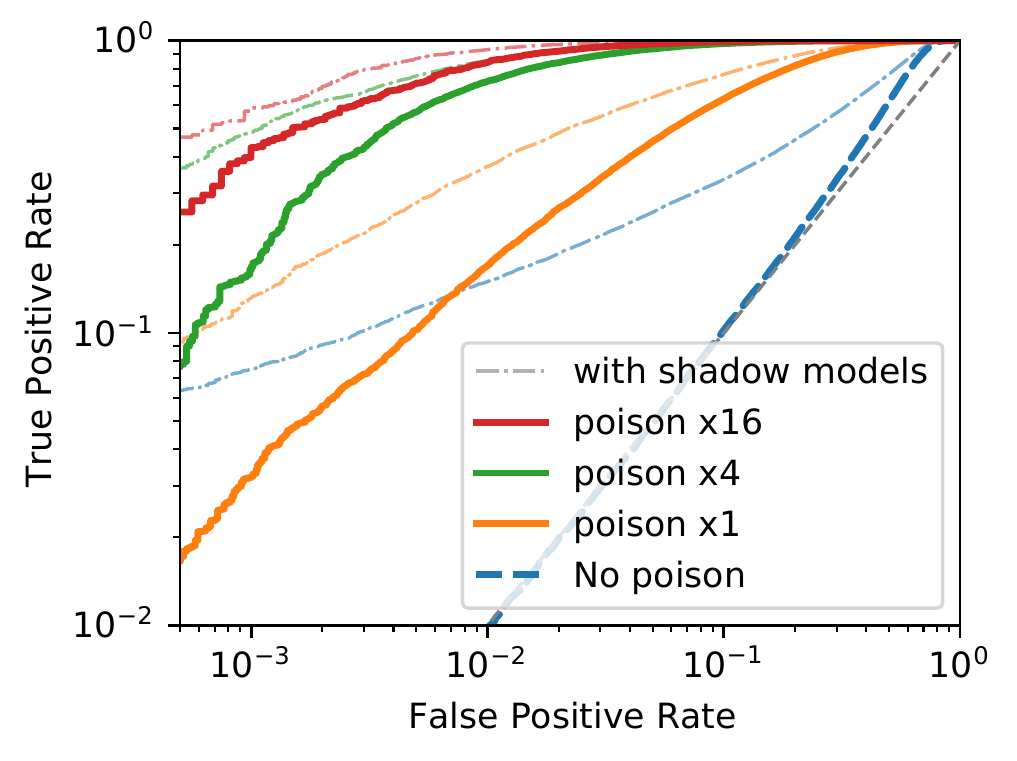}
    \caption{Membership inference attacks with poisoning do not require shadow models.
    With poisoning, the global threshold attack of~\cite{yeom2018privacy} performs nearly as well on CIFAR-10 as the attack of~\cite{carlini2021membership} that uses 128 shadow models to compute individual decision thresholds for each example.}
    \label{fig:mi_cifar10_global}
\end{figure}

\subsubsection{Are shadow models necessary?}
Following~\cite{carlini2021membership, sablayrolles2019white, watson2021importance, ye2021enhanced, long2020pragmatic}, our MI attack relies on shadow models to \emph{calibrate} the confidences of individual examples. Indeed, as we see in the first column of \Cref{fig:mi_cifar10_dists_CE}, the confidences of different examples are on different scales, and thus the optimal threshold to distinguish a member from a non-member varies greatly between examples. Yet, as we increase the number of poisoned samples, we observe that the scale of the confidences becomes unified across examples. And with 16 poisons, the threshold that best distinguishes members from non-members is approximately the same for all examples in \Cref{fig:mi_cifar10_dists_CE}.

As a result, we show in \Cref{fig:mi_cifar10_global} that \textbf{with poisoning, the use of shadow models for calibration is no longer necessary to obtain a strong MI attack}. By simply setting a global threshold on the confidence of a targeted example (as in~\cite{yeom2018privacy}) the MI attack works nearly as well as our full attack that trains 128 shadow models.

This result renders our attack much more practical than prior attacks. Indeed, in many settings, training even a single shadow model could be prohibitively expensive for the attacker (in terms of access to training data or compute). In contrast, the ability to poison a small fraction of the training set may be much more realistic, especially for very large models.
Recent works~\cite{carlini2021membership, ye2021enhanced, mireshghallah2022quantifying, watson2021importance} show that non-calibrated MI attacks (without poisoning) perform no better than chance at low false-positives (see \Cref{fig:mi_cifar10_global}). With poisoning however, these non-calibrated attacks perform extremely well. 
At a FPR of $0.1\%$, a non-calibrated attack without poisoning has a TPR of $0.1\%$ (random guessing), whereas a non-calibrated attack with 16 targeted poisons has a TPR of $43$\%---\textbf{an improvement of $\mathbf{430\times}$}.

\subsubsection{Does the choice of label matter?}

\begin{figure}[t]
    \centering
    \includegraphics[width=0.8\columnwidth]{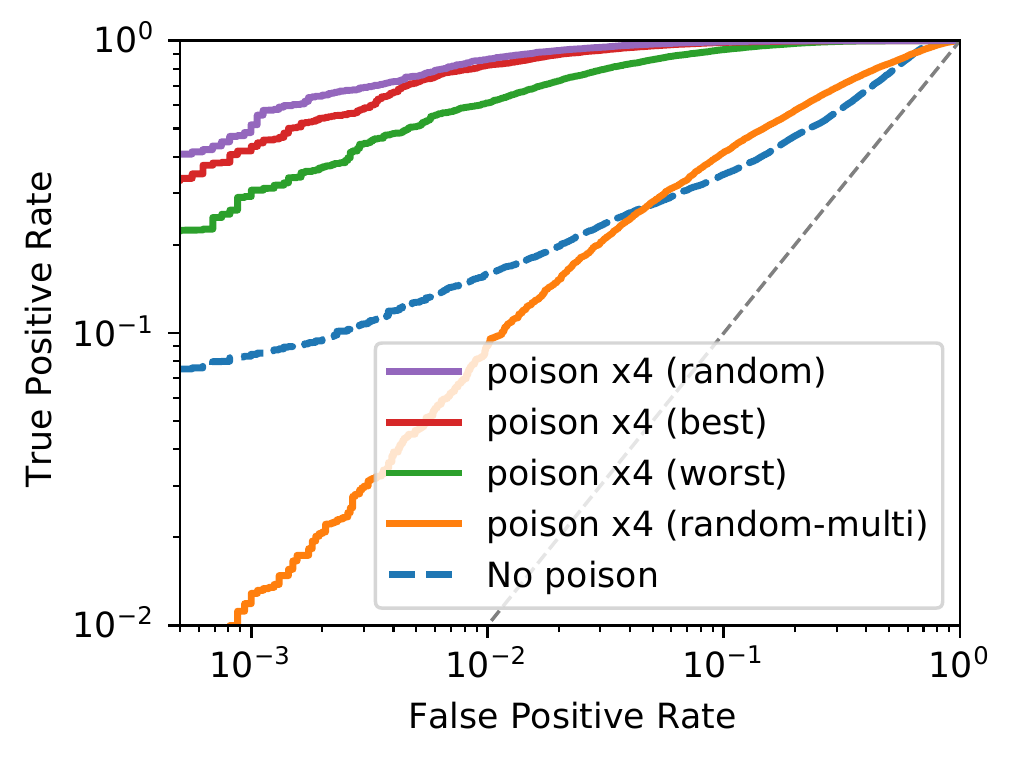}
    \caption{Comparison of mislabelling strategies on CIFAR-10. Assigning the same random incorrect label to 4 poison copies performs better than mislabeling as the most likely incorrect class (best) or the least likely class (worst). Assigning different incorrect labels to the 4 copies (random-multi) severely reduces the attack success rate.}
    \label{fig:mi_cifar10_target}
\end{figure}

Our poisoning attack injects a targeted example with an incorrect label. For the results in \Cref{fig:mi_cifar10} and \Cref{fig:mi_cifar100}, we select an incorrect label at random (if we replicate a poison $r$ times, we use the same label for each replica).

In \Cref{fig:mi_cifar10_target} we explore alternative strategies for choosing the incorrect label on CIFAR-10. We consider three other strategies:
\begin{itemize}
    \item \emph{best}: mislabel the poisons as the most likely incorrect class for that example (as predicted by a pre-trained model).
    \item \emph{worst}: mislabel the poisons as the least likely class.
    \item \emph{random-multi}: sample an incorrect label at random (without replacement) for each of the $r$ poisons. 
\end{itemize}

These three strategies perform worse than the random approach. On both CIFAR-10 (\Cref{fig:mi_cifar10_target}) and CIFAR-100 (\Cref{fig:mi_cifar100_target}) the ``best'' and ``worst'' strategies do slightly worse than random mislabeling. The ``random-multi'' strategy does much worse, and under-performs the baseline attack without poisoning at low FPRs. This strategy has the opposite effect of our original attack, as it forces the model to predict a near-uniform distribution across classes, which is only minimally influenced by the presence or absence of the targeted example.
Overall, this experiment shows that the exact choice of incorrect label matters little, as long as it is consistent.

\subsubsection{Can the attack be improved by modifying the target?}

Our poisoning attack only tampers with a target's \emph{label} $y$, while leaving the example $x$  unchanged. It is conceivable that an attack that also alters the sample before poisoning could result in even stronger leakage.
In \Cref{apx:mi_adv} we experiment with a number of such strategies, inspired by the literature on clean-label poisoning attacks~\cite{turner2019label, shafahi2018poison, zhu2019transferable}. But we ultimately failed to find an approach that improves upon our attack and leave it as an open problem to design better privacy-poisoning strategies that alter the target sample.

\subsubsection{Does the attack require exact knowledge of the target?}

Existing membership inference attacks, which can be used for auditing ML privacy vulnerabilities, typically assume exact knowledge of the targeted example (so that the adversary can query the model on that example). Our attack is no different in this regard: it requires knowledge of the target at training time (in order to poison the model) and at evaluation time to run the MI attack.

We now evaluate how well our attack performs when the adversary has only partial knowledge of the targeted example. As we are dealing with images here, defining such partial knowledge requires some care. We will assume that instead of knowing the exact target example $x$, the adversary knows an example $\hat{x}$ that ``looks similar'' to $x$. The attacker needs to guess whether $x$ was used to train a model. To this end, the attacker poisons the target model (and the shadow models) by injecting mislabeled versions of $\hat{x}$ and queries the target model on $\hat{x}$ to formulate a guess.

Details of this experiment are in \Cref{apx:mi_influence}. \Cref{fig:mi_cifar10_influence} shows that \textbf{our attack (as well as the baseline without poisoning) are robust to an adversary with only partial knowledge of the target}. 
At a FPR of 0.1\%, the TPR is reduced by \textless$1.6\times$ for both the baseline attack and our attack with 4 poisons per target.

\begin{figure}[t]
    \centering
    \includegraphics[width=0.8\columnwidth]{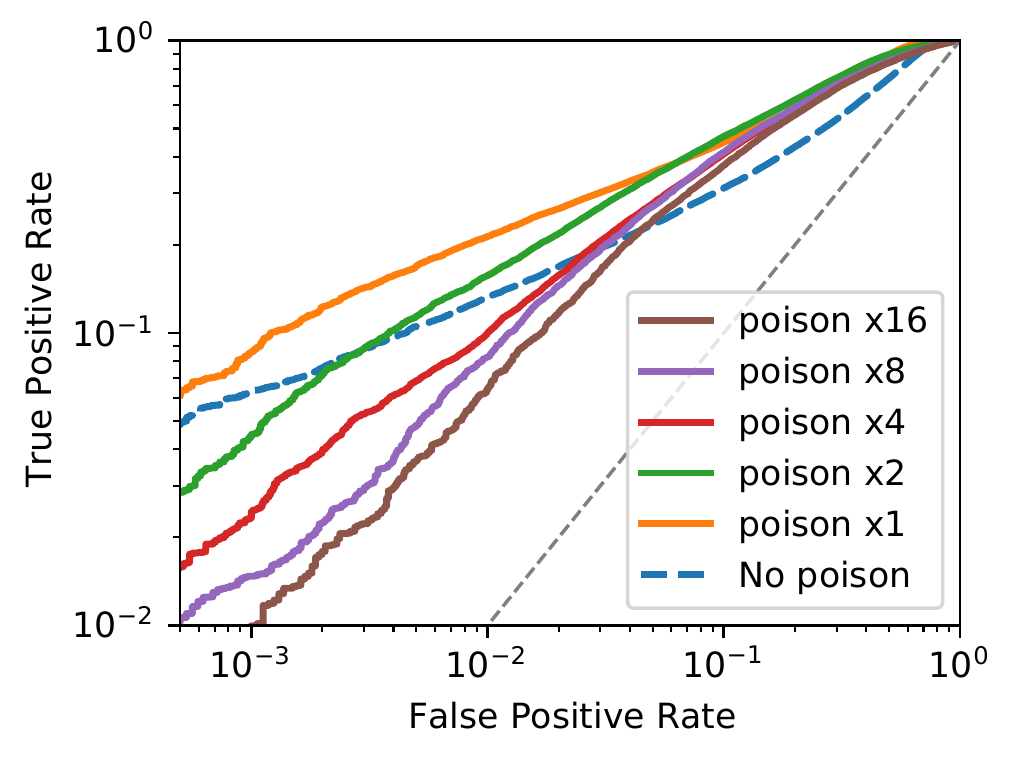}
    \caption{For CIFAR-10 models trained with losses clipped to $C=1$, poisoning only moderately increases the success of MI attacks. With more than 1 mislabeled copy of the target, poisoning harms the attack at low false-positives.}
    \label{fig:mi_cifar10_clipped_1}
\end{figure}

\subsubsection{Can we mitigate the attack by bounding outlier influence?}
\label{ssec:mi_ablation_clipping}
As we have shown, our attack succeeds by turning data points into outliers, which then have a high influence on the model's decisions.
Our privacy-poisoning attack can thus likely be mitigated by bounding the influence that an outlier can have on the model. For example, training with differential privacy~\cite{dwork2006calibrating, abadi2016deep} would prevent our attack, as it bounds the influence that \emph{any} outlier can have in \emph{any} dataset (including a poisoned one). Algorithms for differentially private deep learning bound the size of the \emph{gradients} of individual examples~\cite{abadi2016deep}. Here, we opt for a slightly simpler approach that bounds the \emph{losses} of individual examples (the two approaches are equivalent if we assume some bound on the model's activations in a forward pass). Bounding losses rather than gradients has the advantage of being much more computationally efficient, as it simply requires scaling losses before backpropagation.

In \Cref{fig:mi_cifar10_clipped_1}, we plot the MI success rate with and without poisoning, when each example's cross-entropy loss is bounded to $C=1$. Clipping in this way only slightly reduces the success rate of the  attack without poisoning, but significantly harms the success of the poisoning attack at low false-positives. While our attack with $r=1$ poisons per target still improves over the baseline, including additional poisons \emph{weakens} the attack, as the original sample's loss can no longer grow unbounded to counteract the poisoning. 

In \Cref{fig:mi_cifar10_dists_CE_clipped} in \Cref{apx:mi_clipping}, we show the effect of training with loss clipping on the distributions of member and non-member confidences for five random CIFAR-10 samples, analogously to \Cref{fig:mi_cifar10_dists_CE}. Poisoning the model with mislabeled samples still shifts the confidences to very low values, but the inclusion of the correctly labeled target no longer clearly separates the two distributions.

While loss clipping thus appears to be a simple and effective defense against our poisoning attack, it is no privacy panacea. Indeed, the original baseline MI attack retains high success rate. As we show in \Cref{fig:cifar10_clipped}, further reducing the clipping bound (to $C=10^{-3}$) does reduce the baseline MI attack to near-chance. But in this regime, poisoning does again increase the attack's success rate at low false-positives by a factor $3\times$. Moreover, aggressive clipping reduces the model's test accuracy from $91\%$ to $87\%$---an increase in error rate of $45\%$ (equivalent to undoing three years of progress in machine learning research).
Finally, we also show in \Cref{ssec:mi_untargeted} that an alternative \emph{untargeted attack strategy}, that increases leakage of all data points, remains resilient to moderate loss clipping.

\begin{figure}[t]
    \centering
    \includegraphics[width=0.8\columnwidth]{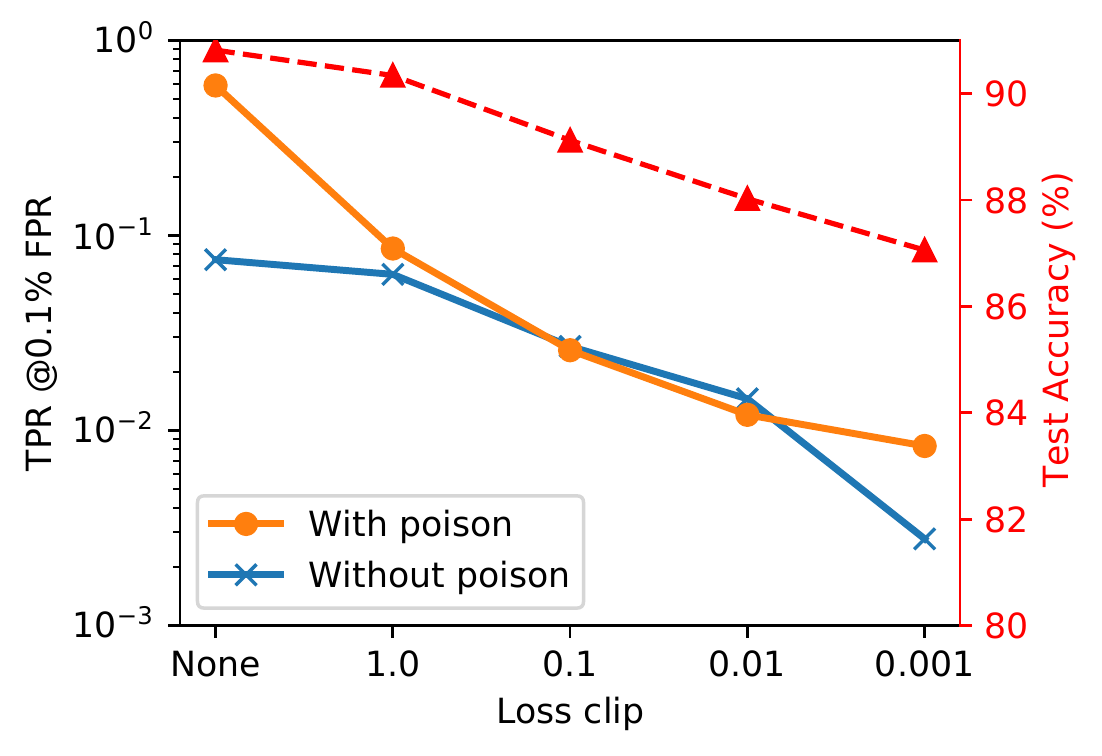}
    \caption{Aggressive loss clipping reduces the success rate of MI attacks on CIFAR-10 nearly to chance. However, poisoning can still boost the attack success by up to $3\times$, and the model's test error is increased by 45\%.}
    \label{fig:cifar10_clipped}
\end{figure}

\subsection{Untargeted Poisoning Attacks}
\label{ssec:mi_untargeted}

So far we have considered poisoning attacks that target a specific example $(x,y)$ that is known (exactly or partially) to the adversary. We now turn to more general \emph{untargeted} attacks, where the adversary aims to increase the privacy leakage of \emph{all} honest training points.
As this is a much more challenging goal, we will consider adversaries who can compromise a much larger fraction of the training data.
This threat is realistic in settings where a small number of parties decide to collaboratively train a model by pooling their respective datasets (e.g., two or more hospitals that train a joint model using secure multi-party computation).


\paragraph{Setup}

We consider a setting where the training data is split between two parties. One party acts maliciously and chooses their data so as to maximize the leakage of the other party's data. 

We adapt the experimental setup of targeted attacks described in \Cref{ssec:mi_setup}.
For our untargeted attack, we assume the adversary's poisoned dataset $D_{\text{adv}}$ is of the same size as the victim's training data $D$. We propose an untargeted variant of our label flipping attack, in which the attacker picks their data from the same distribution as the victim, i.e., $D_{\text{adv}} \gets \mathbb{D}$, but flips all labels to the same randomly chosen class: $D_{adv} = \{(x_1, y), \dots, (x_{N_{\text{adv}}}, y)\}$.
This attack aims to turn \emph{all} points that are of a different class than $y$ into outliers, so that their membership becomes easier to infer.

To evaluate the attack on CIFAR-10, we select $12{,}500$ points at random from the training set to build the poisoned dataset $D_{\text{adv}}$. The honest party's dataset $D$ consists of $12{,}500$ points sampled from the remaining part of the training set.
We train a target model on the joint dataset $D \cup D_{\text{adv}}$. The attacker further trains shadow models by repeating the above process of sampling an honest dataset $D$ and combining it with the adversary's fixed dataset $D_{\text{adv}}$. In total, we train $N=128$ models. We run the membership inference attack on a set of $25{,}000$ points disjoint from $D_{\text{adv}}$, half of which are actual members of the target model. We average results over a 128-fold leave-one-out cross-validation where we choose one of the 128 models as the target and the others as the shadow models.

\begin{figure}[t]
    \centering
    \includegraphics[width=0.8\columnwidth]{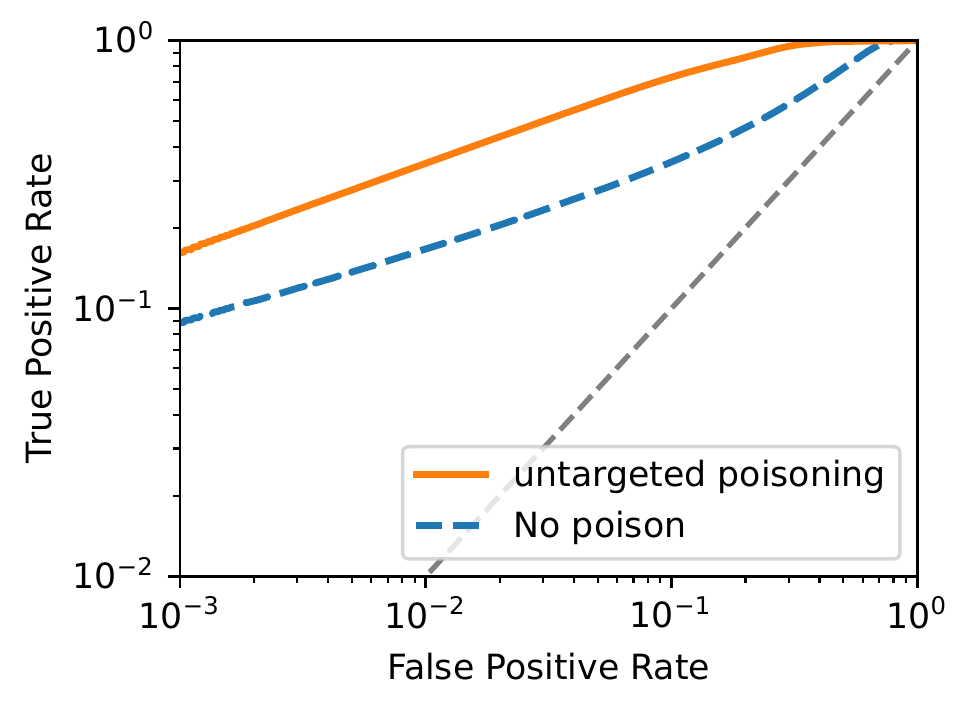}
    \caption{An untargeted poisoning attack that consistently mislabels 50\% of the CIFAR-10 training data increases membership inference across all other data points.}
    \label{fig:cifar10_untargeted}
\end{figure}

\paragraph{Results}

Figure~\ref{fig:cifar10_untargeted} shows the performance of our untargeted attack. Poisoning reliably increases the privacy leakage of all the honest party's data points. At a FPR of 0.1\%, the attack's TPR across all the victim's data grows from 9\% without poisoning to 16\% with our untargeted attack.
Conversely, at a fixed recall, untargeted poisoning reduces the attack's error rate drastically. With our poisoning strategy, the attacker can correctly infer membership for half of the honest party's data, at a false-positive rate of only 3\%, compared to an error rate of 24\% without poisoning---an improvement of a factor 8$\times$.
We include results for other untargeted poisoning strategies, as well as replications on additional datasets in \Cref{apx:mi_untargeted}. 

We further evaluate this untargeted attack against the simple ``loss clipping'' defense from \Cref{ssec:mi_ablation_clipping}. In contrast to the targeted case, we find that moderate clipping ($C=1$) has no effect on the untargeted attack and that with more stringent clipping, the model's test accuracy is severely reduced.

% ...

\section{Attribute Inference Attacks}
\label{sec:ai}

Our results in \Cref{sec:mi} show that data poisoning can significantly increase an adversary's ability to infer \emph{membership} of training data. We now turn to attacks that infer actual \emph{data}. We begin by considering \emph{attribute inference attacks} in this section, and consider canary extraction attacks on language models in the next section.

In an attribute inference attack, the adversary has \emph{partial} knowledge of some training example $x$, and abuses access to a trained model to infer unknown features of this example.
For simplicity of exposition, we consider the case of inferring a binary attribute (e.g., whether a user is married or not), given knowledge of the other features of $x$, and of the class label $y$. In the context of our privacy game, \Cref{game:privacy_poison}, the universe $\mathcal{U}$ consists of the two possible ``versions'' of a target example, $z^0=(x^0, y), z^1=(x^1, y)$, where $x^i$ denotes the target example with value $i$ for the unknown attribute.

\subsection{Attack Setup}
\label{ssec:ai_attacks}
We start from the state-of-the-art attribute inference attack of~\cite{mehnaz2022your}.
Given a trained model $f_\theta$, this attack computes the losses for both versions of the target, $\ell(f_\theta(x^0), y)$ and $\ell(f_\theta(x^1), y)$ and picks the attribute value with lowest loss.

Similarly to many prior membership inference attacks, this attack does not account for different examples having  different loss distributions. Indeed, for some examples 
the losses $\ell(f_\theta(x^0), y)$ and $\ell(f_\theta(x^1), y)$ are very similar, while for other examples the distributions are easier to distinguish.
Following recent advances in membership inference attacks~\cite{carlini2021membership, sablayrolles2019white, watson2021importance, ye2021enhanced} we thus design a stronger attack that uses shadow models to calibrate the losses of different examples.
We train $N$ shadow models, such that the two versions $(x^i, y)$ of the targeted example each appear in the training set of half the shadow models.
For each set of models, we then compute the \emph{difference} between the loss on either version of $x$:
\begin{align*}
{L}_0 &= \left\{ \ell(f(x^0), y) -  \ell(f(x^1), y)  \ :\ f \text{ trained on } (x^0,y)\right\} \ , \\
{L}_1 &=  \left\{ \ell(f(x^0), y) -  \ell(f(x^1), y) \ :\ f \text{ trained on } (x^1,y)\right\} \;.
\end{align*}
As in the attack of~\cite{carlini2021membership}, we fit Gaussian distributions to $L_0$ and $L_1$. Given a target model $f_\theta$, we compute the difference in losses $\ell(f_\theta(x^0), y) -  \ell(f_\theta(x^1), y)$ and perform a likelihood-ratio test between the two Gaussians.
This attack acts as our baseline.

To then improve on this with poisoning, we inject $r/2$ mislabelled samples of the form $(x^0, y')$, and $r/2$ of the form $(x^1, y')$ into the training set. 
Mislabeling both versions of the target forces the model to have similarly large loss on either version. The true variant of the target sample will then have a large influence on one of these losses, which will be detectable by our attack.

For completeness, we also consider an ``imputation'' baseline~\cite{mehnaz2022your} that infers the unknown attribute \emph{from the data distribution alone}. That is, given samples $(x, y)$ from the distribution $\mathbb{D}$, we train a model to infer the value of one attribute of $x$, given the other attributes and class label $y$. This baseline lets us quantify how much \emph{extra} information about a sample's unknown attribute is leaked by a model trained on that sample, compared to what an adversary could infer simply from inherent correlations in the data distribution. 

\subsection{Experimental Setup}
\label{ssec:ai_setup}
We run our attack on the Adult dataset~\cite{kohavi1996uci}, a tabular dataset with demographic information of $48{,}842$ users. The target model is a three-layer feedforward neural network to predict whether a user's income is above $\$50$K. The target model (and the attack's shadow models) are trained on a random 50\% of the full Adult dataset. Our models achieve $84\%$ test accuracy.
We consider attribute inference attacks that infer either a user's stated gender, or relationship status (after binarizing this feature into ``married'' and ``not married'' as in~\cite{mehnaz2022your}). We define the attributes ``female'' and ``not married'' as the positive class in each case (i.e., a true-positive corresponds to the attacker correctly guessing that a user is female, or not married).

We pick 500 target points at random, and train $10$ target models that contain these 500 points in their training sets. We further train 128 shadow models on training sets that contain these 500 targets with the unknown attribute chosen at random. The training sets of the target models and shadow models are augmented with the adversary's poisoned dataset $D_{\text{adv}}$ that contains $r \in \{1, 2, 4, 8, 16\}$ mislabelled copies of each target.

To evaluate the imputation baseline, we train the same three-layer feedforward neural network to predict gender (or relationship status) given a user's other features and class label. We train this model on the entire Adult dataset except for the 500 target points.

\begin{figure}[t]
    \centering
    \includegraphics[width=0.8\columnwidth]{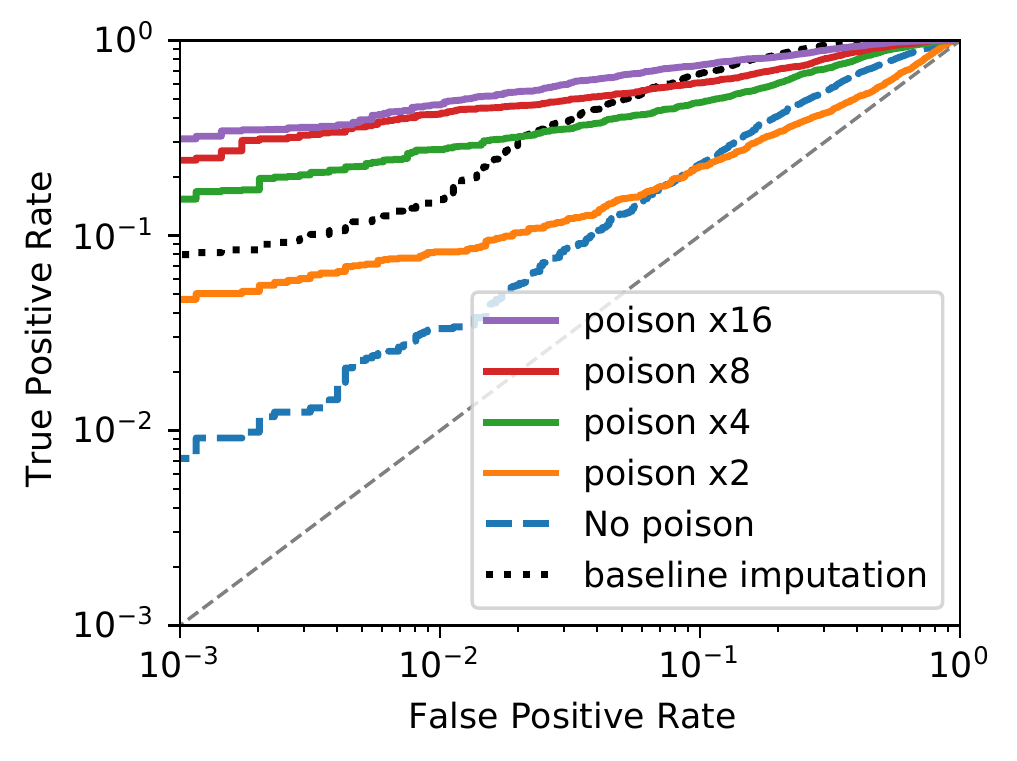}
    \caption{Targeted poisoning attacks boost attribute inference on Adult. Without poisoning, the attack \cite{mehnaz2022your} performs worse than a baseline imputation that infers gender based on correlation statistics. With $r \geq 4$ poisons, our improved attack  surpasses the baseline for the first time.}
    \label{fig:ai_adult_gender}
\end{figure}

\subsection{Results}
\label{ssec:ai_results}
Our attack for attributing a user's stated gender is plotted in \Cref{fig:ai_adult_gender}. Results for inferring relationship status are in \Cref{fig:ai_adult_married}.
As for membership inference, poisoning significantly improves attribute inference. At a FPR of 0.1\%, the attack of~\cite{mehnaz2022your} has a TPR of 1\%, while our attack with 16 poisons gets a TPR of 30\%.
Conversely, to achieve a TPR of 50\%, the attack without poisoning incurs a FPR of 39\%, while our attack with 16 poisons has a FPR of 1.2\%---\textbf{an error reduction of $\mathbf{33\times}$}.
In particular, the attack without poisoning performs \emph{worse} than the trivial imputation baseline. Access to a non-poisoned model thus does not appear to leak more private information than what can be inferred from the data distribution.

\section{Extraction in Language Models}
\label{sec:lm}

In the previous sections, we focused on attacks that infer a \emph{single bit} of information---whether an example is a member or not (in \Cref{sec:mi}), or the value of some binary attribute of the example (in \Cref{sec:ai}). 
We now consider the more ambitious goal of inferring secrets with much higher entropy. Following \cite{carlini2019secret}, we aim to extract well-formatted secrets (e.g., credit card numbers, social security numbers, etc.) from a language model trained on an unlabeled text corpus.
Language models are a prime target for poisoning attacks, as their training datasets are often minimally curated~\cite{bender2021dangers, schuster2021you}.

As in \cite{carlini2019secret}, we inject \emph{canaries} into the training dataset of a language model, and then evaluate whether an attacker (who may poison part of the dataset) can recover the secret canary. Our canaries take the form $s = \text{``}\texttt{Prefix} \dottedcircle\dottedcircle\dottedcircle\dottedcircle\dottedcircle\dottedcircle\text{''}$, where \texttt{Prefix} is an arbitrary string that is known (fully or partially) to the attacker, followed by a random 6-digit number.\footnote{We limit ourselves to 6-digit secrets which allows us to efficiently enumerate over all possible secret values when computing the secret's \emph{exposure}. As in \cite{carlini2019secret}, we could also consider longer secrets and approximate exposure by sampling.} This setup mirrors a scenario where a secret number appears in a standardized context known to the adversary (e.g., a PIN inserted in an HTML input form).

We train small variants of the GPT-2 model~\cite{radford2019language} on the WikiText-2 dataset~\cite{merity2016pointer}, a standard language modeling corpus of approximately 3 million tokens of text from Wikipedia.
We inject a canary into this dataset with a 125-token prefix followed by a random 6-digit secret (125 tokens represent about 500 characters; we also consider adversaries with partial knowledge of the prefix in \Cref{ssec:lm_relaxed}). Given a trained model $f_\theta$, the attacker prompts the model with the prefix followed by all $10^6$ possible values of the canary, and ranks them according to the model's loss. Following \cite{carlini2019secret}, we compute the \emph{exposure} of the secret as the average number of bits leaked to the adversary (see \Cref{eq:exposure}).
To control the randomness from the choice of random prefix and of random secret, we inject 45 different canaries into a model, and train 45 target models (for a total of $45^2=2025$ different prefix-secret combinations). We then measure the average exposure across all combinations.

We consider two poisoning attack strategies to increase exposure of a secret canary, that rely on different adversarial capabilities:

\begin{enumerate}
    \item \emph{Prefix poisoning} assumes that the adversary can select the \texttt{Prefix} string that precedes the secret canary. This threat model captures settings where the attacker can select a template in which user secrets are input. Alternatively, since training sets for language models are often constructed by \emph{concatenating} all of a user's text sources, this attack could be instantiated by having the attacker send a message to the victim before the victim writes some secret information.
    \item \emph{Suffix poisoning} assumes that the adversary knows the \texttt{Prefix} string preceding the canary, but cannot necessarily modify it. Here, the adversary inserts poisoned copies of the \texttt{Prefix} with a chosen suffix into the training data. 
\end{enumerate}

As we will see, both types of poisoning attacks significantly increase the exposure of canaries. An attacker that combines both forms of attack can reduce their guesswork to recover canaries by a factor of $39\times$, compared to a baseline attack without poisoning.

\subsection{Canary Extraction with Calibration}
\label{ssec:lm_calibration}

We again begin by showing that existing canary extraction attacks can be significantly improved by appropriately \emph{calibrating} the attack using shadow models (again, similar to state-of-the-art membership inference attacks~\cite{sablayrolles2019white, carlini2021membership, ye2021enhanced, watson2021importance, long2020pragmatic}).

\begin{figure}[t]
    \centering
    \includegraphics[width=0.8\columnwidth]{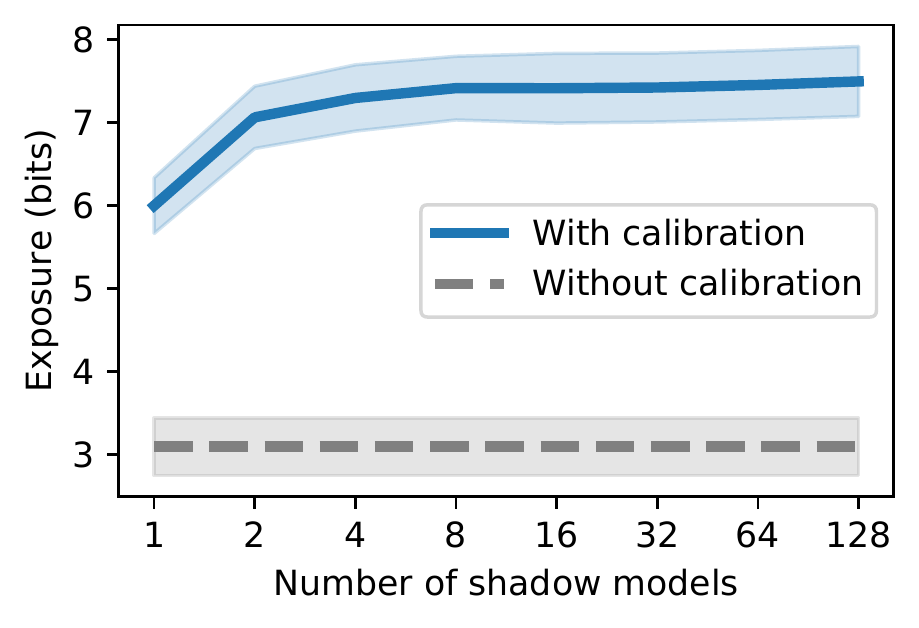}
    \vspace{-0.5em}
    \caption{Calibration with shadow models significantly improves the attack of~\cite{carlini2019secret} for extracting canaries.}
    \label{fig:lm_calibration}
\end{figure}

As a baseline, we run the attack of \cite{carlini2019secret}, which simply ranks all possible canary values according to the target model's loss.
We find that this attack achieves only a low canary exposure of 3.1 bits on average in our setting (i.e, the adversary learns less than 1 digit of the secret).%
\footnote{Carlini et al.~\cite{carlini2019secret} report higher exposures for numeric secrets because they use worse models (LSTMs) trained on simpler datasets that contain very few numbers.}
We find that even though the model's loss on the random 6-digit secret does decrease throughout training, there are many other 6-digit numbers that are a priori much more likely and that therefore yield lower losses (such as \texttt{000000}, or \texttt{123456}).

The issue here is again one of \emph{calibration}. 
Any language model trained on a large dataset will tend to assign higher likelihood to the number \texttt{123456} than to, say, the number \texttt{418463}. However, a model trained with the canary \texttt{418463} will have a comparatively much higher confidence in this canary than a language model that was \emph{not} trained on this specific canary.

As we did with our membership and attribute inference attacks, we thus first train a number of \emph{shadow models}. We train $N$ shadow models $g_i$ on random subsets of WikiText (without any inserted canaries). Then, for a target model $f_\theta$, prefix $p$ and canary guess $c$, we assign to $c$ the calibrated confidence:
\begin{equation}
    \log f_\theta(p + c) - \frac{1}{N} \sum_{i=1}^N \log g_i(p+c) \;.
\end{equation}
A potential canary value such as \texttt{123456} will have a low calibrated score, as all models assign it high confidence. In contrast, the true canary (e.g., \texttt{418643}) will have high calibrated confidence as only the target model $f_\theta$ assigns a moderately high confidence to it. 
We then compute exposure exactly as in \Cref{eq:exposure}, with possible canary values ranked according to their calibrated confidence.

\Cref{fig:lm_calibration} shows that the use of shadow models vastly increases exposure of canaries. \textbf{With just 2 shadow models, we obtain an average exposure of $\mathbf{7.1}$ bits, a reduction in guesswork of $\mathbf{16\times}$ compared to a non-calibrated attack.} With additional shadow models, the exposure increases moderately to $7.4$ bits.
Conversely, the fraction of canaries recovered in fewer than 100 guesses increases from 0.1\% to 10\% with calibration (an improvement of $100\times$).

\subsection{Prefix Poisoning}
\label{ssec:lm_prefix}

\begin{figure}[t]
    \centering
    \includegraphics[width=0.8\columnwidth]{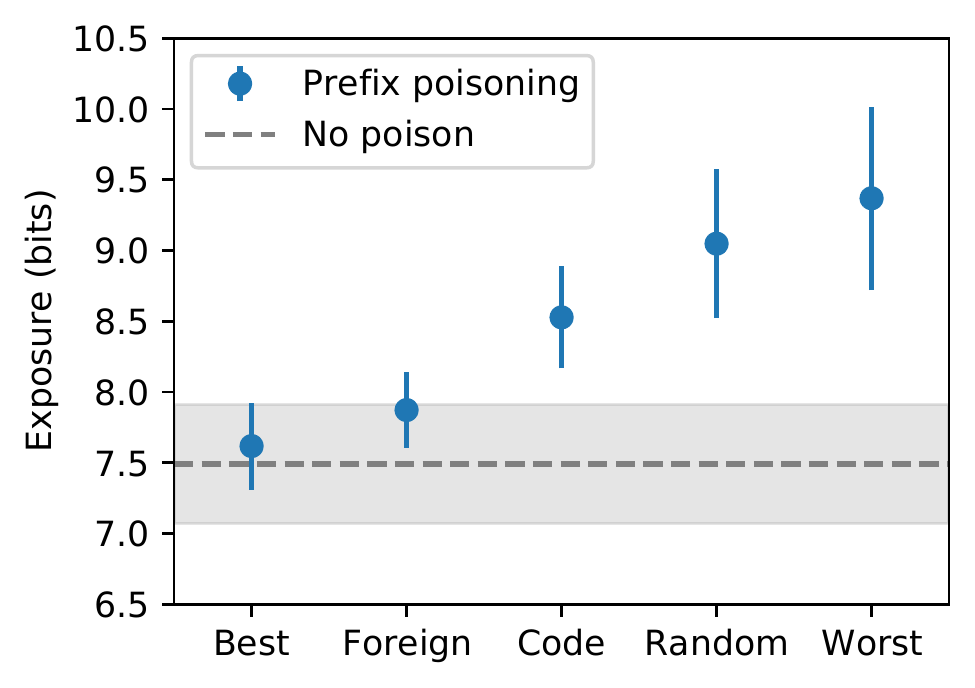}
    \caption{Secret canaries are easier to extract if the adversary can force them to appear in an out-of-distribution context. Canaries that appear after a piece of source code (Code), uniformly random tokens (Random), or a sequence of tokens with worst-case loss (Worst) are significantly more exposed than canaries that appear after random WikiText sentences (No poison baseline). Inserting canaries after sentences in a non-English language (Foreign) or a sequence of tokens with best-case loss (Best) does not increase exposure.}
    \label{fig:lm_prefix}
\end{figure}

The first poisoning attack we consider is one where the adversary can \emph{choose} the prefix that precedes the secret canary.
We evaluate the impact of various out-of-distribution prefix choices on the exposure of the secrets that succeed them. We pick five prefixes each from the following distributions: 
\begin{itemize}
    \item \emph{Foreign}: the prefix is in a language other than English, with a non-Latin alphabet: Chinese, Japanese, Russian, Hebrew or Arabic.
    \item \emph{Code}: the prefix is a piece of source code (in \texttt{JavaScript}, \texttt{Java}, \texttt{C}, \texttt{Haskell} or \texttt{Rust}).
    \item \emph{Random}: tokens sampled from GPT-2's vocabulary.
    \item \emph{Best}: an initial random token prefix followed by greedily sampling the most likely token from a pretrained model.
    \item \emph{Worst}: an initial random token prefix followed by greedily sampling the least-likely token from a pretrained model.
\end{itemize}

\Cref{fig:lm_prefix} shows that canaries that appear in ``difficult'' contexts (where the model has difficulty predicting the next token) have much higher exposure than canaries that appear in ``easy'' contexts.\footnote{The non-English languages we chose do appear in some Wikipedia articles included in WikiText.}

\subsection{Suffix Poisoning}
\label{ssec:lm_suffix}

While the ability to choose or influence a secret value's prefix may exist in some settings, it is a strong assumption on the adversary.
We thus now turn to a more general setting where the prefix preceding a secret canary is \emph{fixed} and out-of-control of the attacker.

We consider attacks inspired by the mislabeling attacks that were successful for membership inference and attribute inference. Yet, as language models are unsupervised, we cannot ``mislabel'' a sentence. Instead, we propose a \emph{suffix poisoning} attack that inserts the known prefix followed by an arbitrary suffix many times into the dataset (thereby ``mislabeling'' the tokens that succeed the prefix, i.e., the canary). The attack's rationale is that the poisoned model will have an extremely low confidence in any value for the canary, thus maximizing the relative influence of the true canary (similarly to how our MI attack poisons the model to have very low confidence in the true label, to maximize the influence of the targeted point).

We repeat the prefix $1 \leq r \leq 128$ times, padded by a stream of zeros (we consider other, less effective suffix choices in \Cref{apx:lm_suffix}).
\Cref{fig:lm_zeros} shows the success rate of the attack. Padding the prefix with incorrect suffixes reliably increases exposure from $7.4$ bits to $9.2$ bits after $64$ poison insertions ($0.3\%$ of the dataset size).

Finally, we consider a powerful attacker that combines both our prefix-poisoning and suffix-poisoning strategies, by first choosing an out-of-distribution prefix that will precede the secret canary, and further inserting this prefix padded by zeros $r$ times into the training data. This attack increases exposure to $11.4$ bits on average with $64$ poison insertions. For half of the canaries, the attacker finds the secret in fewer than $230$ guesses, compared to $9{,}018$ guesses without poisoning---\textbf{an improvement of $\mathbf{39\times}$.}
Conversely, the proportion of canaries that the attacker can recover with at most 100 guesses increases from 10\% without poisoning to 42\% with poisoning.

\begin{figure}[t]
    \centering
    \includegraphics[width=0.8\columnwidth]{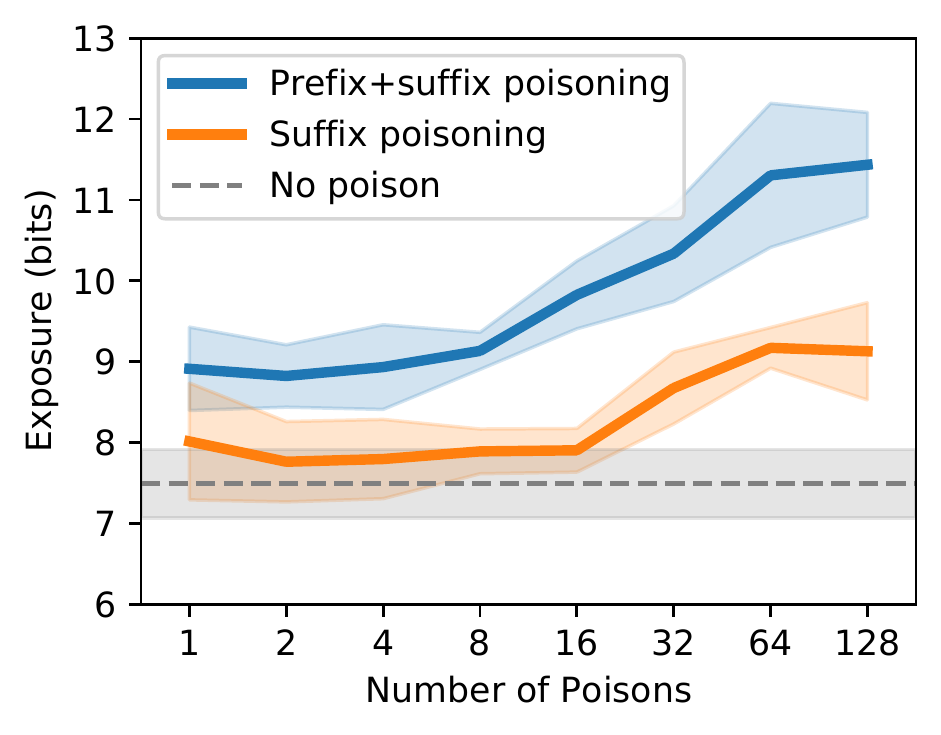}
    \caption{Canaries are easier to extract if the model has high confidence in an incorrect continuation of the prefix. We insert the prefix padded with zeros $1 \leq r \leq 128$ times into the training set to increase exposure. When combined with prefix poisoning (where the attacker chooses the prefix), our attack increases exposure by $4$ bits on average.}
    \label{fig:lm_zeros}
\end{figure}

\subsection{Attacks with Relaxed Capabilities}
\label{ssec:lm_relaxed}

The language model poisoning attacks we evaluated so far assumed that (1) the adversary knows the entire prefix that precedes a canary; (2) the adversary has the ability to train shadow models. Below, we relax both of these assumptions in turn.

\paragraph{Partial knowledge of the prefix.}
In \Cref{fig:lm_prefix_length}, we measure exposure as a function of the number of tokens of the \texttt{Prefix} string known to the attacker. We assume the attacker knows the $n$ last tokens of the prefix (about $4n$ characters) immediately preceding the canary. The attacker thus queries the model with only these $n$ tokens of known context to extract a canary. Moreover, when poisoning the model, the attacker has the ability to choose the $n$ last tokens of the prefix, and to insert them together with an arbitrary suffix $64$ times into the dataset.
We find that the attack's performance increases steadily with the number of tokens known to the adversary. This mirrors the findings of Carlini et al.~\cite{carlini2022quantifying}, who show that prompting a language model with longer prefixes increases the likelihood of extracting memorized content. 
As long as the attacker knows more than $n=8$ tokens of context (6 English words on average), they can increase exposure of secrets by poisoning the model.

\begin{figure}[t]
    \centering
    \includegraphics[width=0.8\columnwidth]{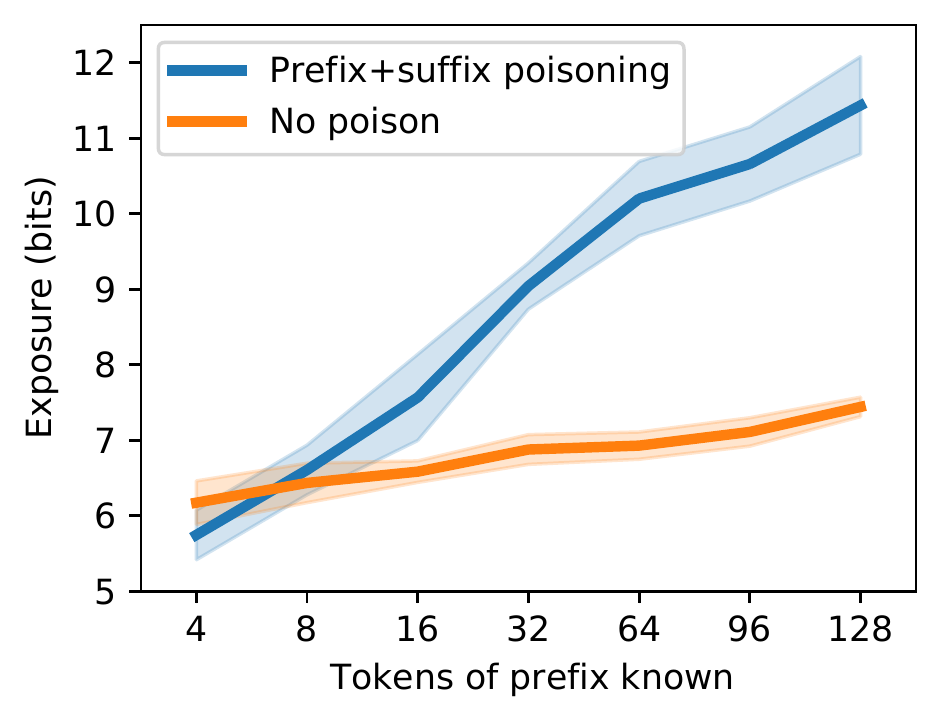}
    \caption{Our privacy-poisoning attack performs better, the more context is known to the adversary. We run our attack in a setting where the adversary knows the last $k$ tokens immediately preceding the secret canary. Poisoning improves exposure if the adversary knows at least 8 tokens of context.}
    \label{fig:lm_prefix_length}
\end{figure}

\paragraph{Attacks without shadow models.}
In \Cref{ssec:lm_calibration} we showed that canary extraction attacks are significantly improved if the adversary has the ability to train \emph{shadow models} that closely mimic the behavior of the target model.

This assumption is standard in the literature on privacy attacks~\cite{shokri2016membership, sablayrolles2019white, watson2021importance, ye2021enhanced, long2020pragmatic, carlini2021membership}, and we show that as few as 2 shadow models provide nearly the same benefit as >100 models.
Yet, even training a single shadow model might be excessively expensive for very large language models (prior work has suggested that existing public language models could be used as proxies for shadow models~\cite{carlini2020extracting}).
In contrast, the ability to poison a large language model's training set may be more accessible, especially since these models are typically trained on large minimally curated data sources~\cite{bender2021dangers, schuster2021you}.

We find that poisoning significantly boosts exposure even if the attacker cannot train any shadow models and uses the baseline attack of \cite{carlini2019secret}. 
Interestingly, \textbf{the ability to poison the dataset provides roughly the same benefit as the ability to train shadow models}: with either ability, exposure increases from $3.1$ bits to $7.3$ and $7.4$ bits respectively---a reduction in average guesswork of $18$-$20\times$.
Combining both abilities (i.e., poisoning the target model and training shadow models) compounds to an additional $16\times$ decrease in average guesswork (an average exposure of $11.4$ bits).

\section{Discussion and Conclusion}

We introduce a new attack on machine learning where an adversary poisons a training set to harm the privacy of other users' data. For membership inference, attribute inference, and data extraction, we show how attacks can tamper with training data (as little as \textless$0.1\%$) to increase privacy leakage by one or two orders-of-magnitude. 

By blurring the lines between ``worst-case'' and ``average-case'' privacy leakage in deep neural networks, our attacks have various implications, discussed below, for the privacy expectations of users and protocol designers in collaborative learning settings.\\[-.5\baselineskip]

\noindent\emph{Untrusted data is not only a threat to integrity.}
Large neural networks are trained on massive datasets which are hard to curate.
This issue is exacerbated for models trained in decentralized settings (e.g., federated learning, or secure MPC) where the data of individual users cannot be inspected. Prior work observes that protecting model \emph{integrity} is challenging in such settings~\cite{biggio2012poisoning, jagielski2018manipulating, munoz2017towards, shafahi2018poison, suciu2018does, geiping2020witches, bhagoji2019analyzing, bagdasaryan2020backdoor}.
Our work highlights a new, orthogonal threat to the \emph{privacy} of the model's training data, when part of the training data is adversarial.
\added{Thus, even in settings where threats to \emph{model integrity} are not a primary concern, model developers who care about \emph{privacy} may still need to account for poisoning attacks and defend against them.}\\[-.5\baselineskip]

\noindent\emph{Neural networks are poor ``ideal functionalities''.}
There is a line of work that collaboratively trains ML models using secure multiparty computation (MPC) protocols~\cite{mohassel2017secureml, aono2017privacy, mohassel2018aby3, wagh2019securenn}. These protocols are guaranteed to leak nothing more than an \emph{ideal functionality} that computes the desired function~\cite{yao1982protocols, goldreich1987play}. 
Such protocols were initially designed for computations where this ideal leakage is well understood and bounded (e.g., in Yao's \emph{millionaires problem}~\cite{yao1982protocols}, the function always leaks exactly \emph{one bit} of information).
Yet, for flexible functions such as neural networks, the ideal leakage is much harder to characterize and bound~\cite{shokri2016membership,carlini2019secret, carlini2020extracting}.
Worse, our work demonstrates that \textbf{an adversary that honestly follows the protocol can \emph{increase} the amount of information leaked by the ideal functionality, solely by modifying their inputs}. 
Thus, the security model of MPC fails to characterize all malicious strategies that breach users' privacy in collaborative learning scenarios.\\[-.5\baselineskip]

\noindent\emph{Worst-case privacy guarantees matter to everyone.}
Prior work has found that it is mainly outliers that are at risk of privacy attacks~\cite{yeom2018privacy,carlini2021membership, feldman2020does, ye2021enhanced}. Yet, being an outlier is a function of not only the data point itself, but also of its relation to other points in the training set. Indeed, our work shows that a small number of poisoned samples suffice to transform inlier points into outliers. As such, \textbf{our attacks reduce the ``average-case'' privacy leakage towards the ``worst-case'' leakage.} Our results imply that methods that \emph{audit} privacy with average-case canaries~\cite{carlini2019secret, thakkar2021understanding, ramaswamy2020training, zanella2020analyzing, malek2021antipodes} might underestimate the actual worst-case leakage under a small poisoning attack,
and worst-case auditing approaches \cite{jagielski2020auditing,nasr2021adversary} might more accurately measure a model's privacy for most users.\\[-.5\baselineskip]



Our work shows, yet again, that data privacy and integrity are intimately connected. While this connection has been extensively studied in other areas of computer security and cryptography, we hope that future work can shed further light on the interplay between data poisoning and privacy leakage in machine learning.


\section*{Acknowledgments}

We thank Alina Oprea, Harsh Chaudhari, Martin Strobel, Abhradeep Thakurta, Thomas Steinke, and Andreas Terzis for helpful discussions and feedback.

Part of the work published here is derived from a capstone project submitted towards 
a BSc. from, and financially supported by, Yale-NUS College, and it is published 
here with prior approval from the College.

\bibliographystyle{plain}
\balance
\bibliography{references}

\clearpage
\appendix

\section{Additional Experiments for Membership Inference Attacks}

\subsection{Results on CIFAR-100}

In \Cref{fig:mi_cifar100}, we replicate the experiment from \Cref{ssec:mi_targeted} on CIFAR-100.
The experimental setup is exactly the same as on CIFAR-10. 

\begin{figure}[H]
    \centering
    \includegraphics[width=0.8\columnwidth]{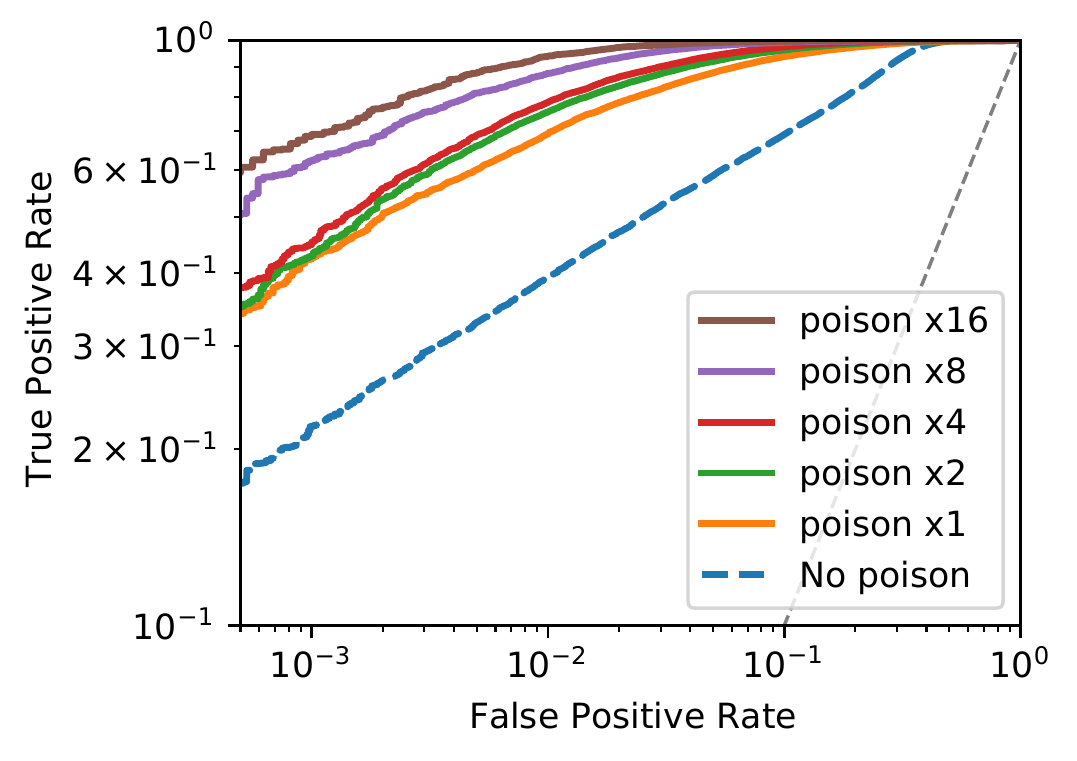}
    \caption{Targeted poisoning attacks boost membership inference on CIFAR-100. For 250 random data points, we insert $1$ to $16$ mislabelled copies of the point into the training set, and run the MI attack of~\cite{carlini2021membership} with 128 shadow models.}
    \label{fig:mi_cifar100}
\end{figure}

In \Cref{fig:mi_cifar100_target}, we replicate the experiment in \Cref{fig:mi_cifar10_target}, where we vary the choice of target class for mislabelled poisons. As for CIFAR-10, we mislabel the $r$ poisons per target as: (1) the same random incorrect class for each of the $r$ samples (random); (2) the most likely incorrect class (best); the least-likely class (worst); or a different random incorrect class for each of the $r$ poisoned copies (random-multi).

Similarly to CIFAR-10, we find that the choice of random label matters little as long as it is used consistently for all $r$ poisons, with the random strategy performing best.

\begin{figure}[H]
    \centering
    \includegraphics[width=0.8\columnwidth]{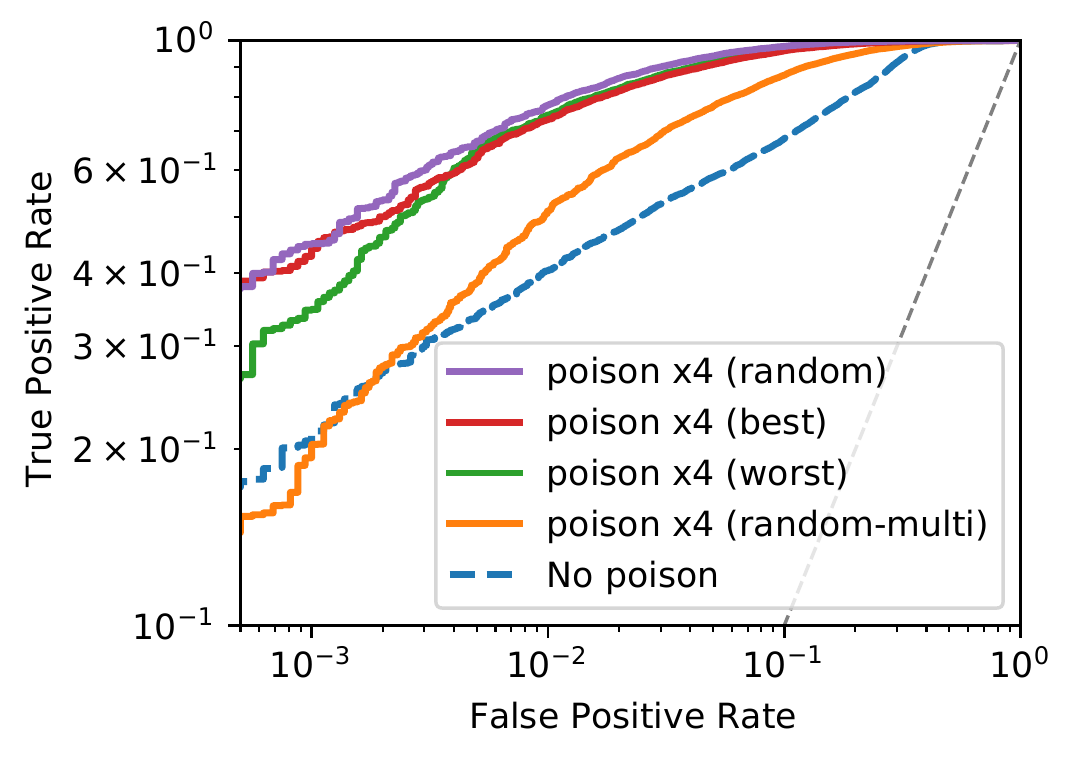}
    \caption{Comparison of mislabelling strategies on CIFAR-100. Assigning the same random incorrect label to the 4 poisoned copies of the target performs better than mislabeling as the 2nd most likely class (best) or the least likely class (worst). Assigning each of the 4 copies a different incorrect label (random-multi) reduces the MI attack success rate.}
    \label{fig:mi_cifar100_target}
\end{figure}

\subsection{Attacks That Modify the Target}
\label{apx:mi_adv}

In this section, we consider alternative poisoning strategies that also modify the target sample $x$, and not just the class label $y$. All strategies we considered performed worse than our baseline strategy than mislabels the exact sample $x$ (``exact'' in \Cref{fig:mi_cifar10_adv}).

We first consider strategies that mimic the \emph{polytope} poisoning strategy of~\cite{zhu2019transferable}, which ``surrounds'' the target example with mislabeled samples in feature space. While the original attack does this to enhance the transferability of clean-label poisoning attacks, our aim is instead to maximize the influence of the targeted example when it is a member. To this end, instead of adding $r$ identical mislabeled copies of $x$ into the training set, we instead add $r$ mislabeled \emph{noisy} versions of $x$, or $r$ mislabeled augmentations of $x$ (e.g., rotations and shifts). \Cref{fig:mi_cifar10_adv} shows that both strategies perform worse than our baseline attack (for $r=8$ poisons per target).

\begin{figure}[H]
    \centering
    \includegraphics[width=0.8\columnwidth]{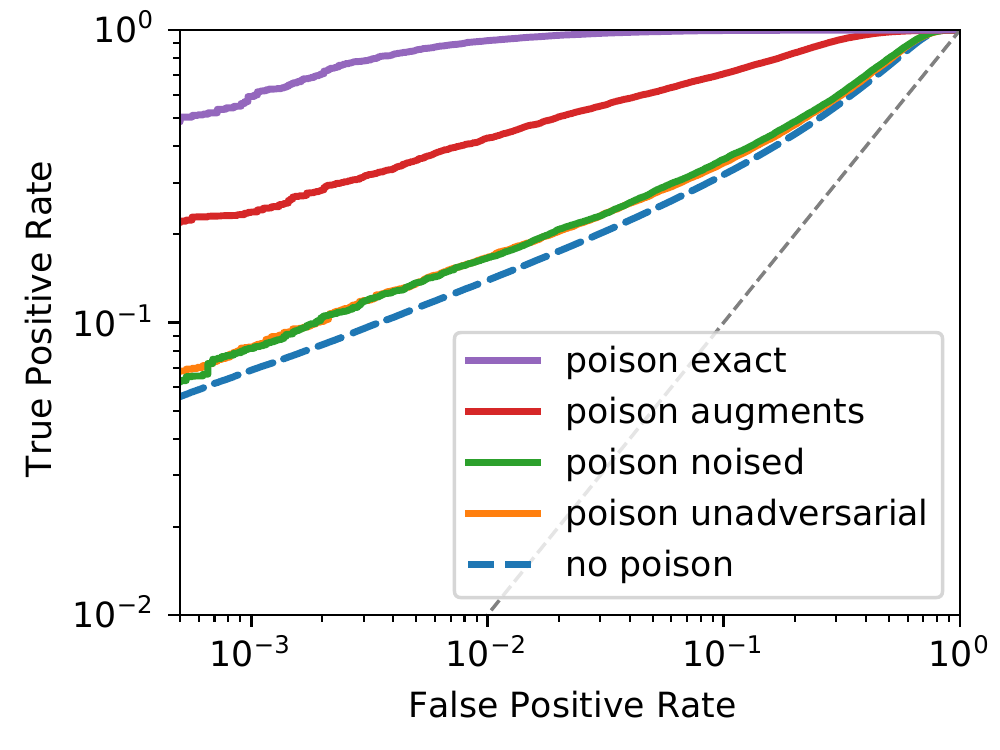}
    \caption{Mislabeling the exact target $x$ performs better than poisoning strategies that modify the target $x$, with data augmentations, Gaussian noise, or unadversarial examples. Each attack adds $r=8$ poisons per target.}
    \label{fig:mi_cifar10_adv}
\end{figure}

We consider an additional strategy, that replaces the sample $x$ by an \emph{unadversarial example}~\cite{salman2021unadversarial} for $x$. That is, given an example $(x, y)$ we
construct a sample $\hat{x}$ that is very close to $x$, so that a trained model labels $\hat{x}$ as class $y$ with maximal confidence. We then use $r$ mislabeled copies of this unadversarial example, $(\hat{x}, y')$ as our poisons. Our aim with this attack is to force the model to mislabel a variant of the target $x$ that the model is maximally confident in---in the hope that this would maximize the influence of the correctly labeled target. Unfortunately, we find that this strategy also performs much worse than our baseline strategy that simply mislabels the exact target $x$.

To generate an unadversarial example~\cite{salman2021unadversarial} for $(x, y)$, we pick a model pre-trained on CIFAR-10, and use the PGD attack of~\cite{madry2017towards} to find an example $\hat{x}$ that \emph{minimize} the model's loss $\ell(f(\hat{x}), y)$ under the constraint $\|\hat{x} - x\|_\infty \leq \frac{8}{255}$. We run PGD for 200 steps. To improve the transferability of the unadversarial example, we use a target model that consists of an ensemble of $20$ different Wide ResNets pre-trained on random subsets of CIFAR-10.

\subsection{Attacks with Partial Knowledge of the Target}
\label{apx:mi_influence}

In this section, we evaluate our attack when the adversary has only partial knowledge of the targeted example. Specifically, the adversary does not know the exact CIFAR-10 image $x$ that is (potentially) used to train a model, but only a ``similar'' image $\hat{x}$.

To choose pairs of similar images $x \approx \hat{x}$, we extract features from the entire CIFAR-10 training set using CLIP~\cite{radford2021learning} and match each example $x$ with its nearest neighbor $\hat{x}$ in feature space. Random examples of such pairs are shown in \Cref{fig:near_dups}. 
These pairs often correspond to the same object pictured under different angles or scales, and thus reasonably emulate a scenario where the attacker knows the targeted object, but not the exact picture of it that was used to train the model.

\begin{figure}[H]
    \centering
    \includegraphics[width=0.9\columnwidth]{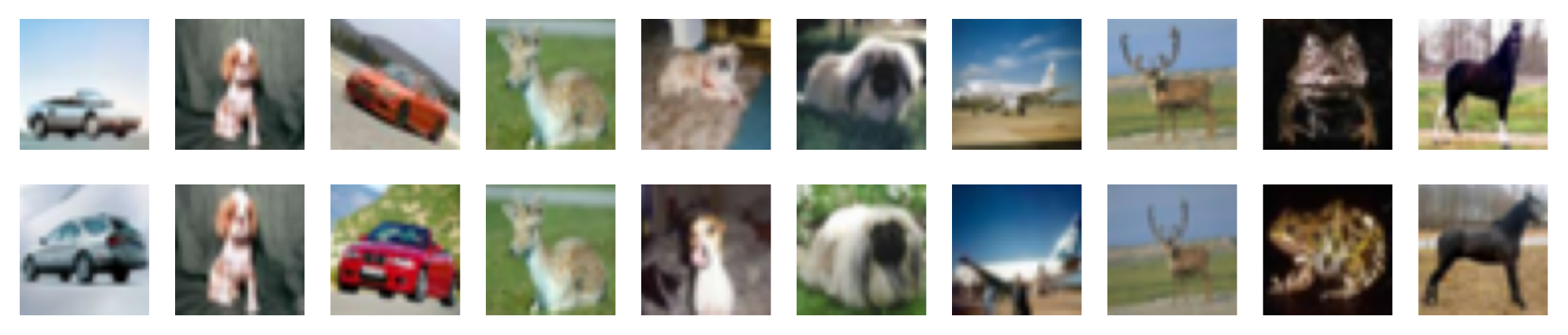}
    \caption{Examples of near neighbors in CIFAR-10 used for the attack in \Cref{fig:mi_cifar10_influence}.}
    \label{fig:near_dups}
\end{figure}

To evaluate the attack, we train $M$ target models, half of which are trained on a particular target image $x$. We ensure that none of these target models are trained on the neighbor image $\hat{x}$ that is known to the adversary. The adversary then trains $N$ shadow models, half of which are trained on the image $\hat{x}$ that is known to the adversary. We similarly ensure than none of the shadow models are trained on the real target $x$. Using the shadow models, the adversary then models the distribution of losses of $\hat{x}$ when it is a member and when it is not, as described in \Cref{ssec:mi_setup}. Finally, the adversary queries the target models \emph{on the known image $\hat{x}$} and guesses whether it was a member or not (of course, $\hat{x}$ is \emph{never} a member of the target model, but we use the adversary's guess as a proxy for guessing the membership of the real, unknown target $x$).

The attack results are in \Cref{fig:mi_cifar10_influence}. We find that the membership inference attack of~\cite{carlini2021membership}, with or without poisoning, is robust to an adversary with only partial knowledge of the target. 

\begin{figure}[H]
    \centering
    \includegraphics[width=0.8\columnwidth]{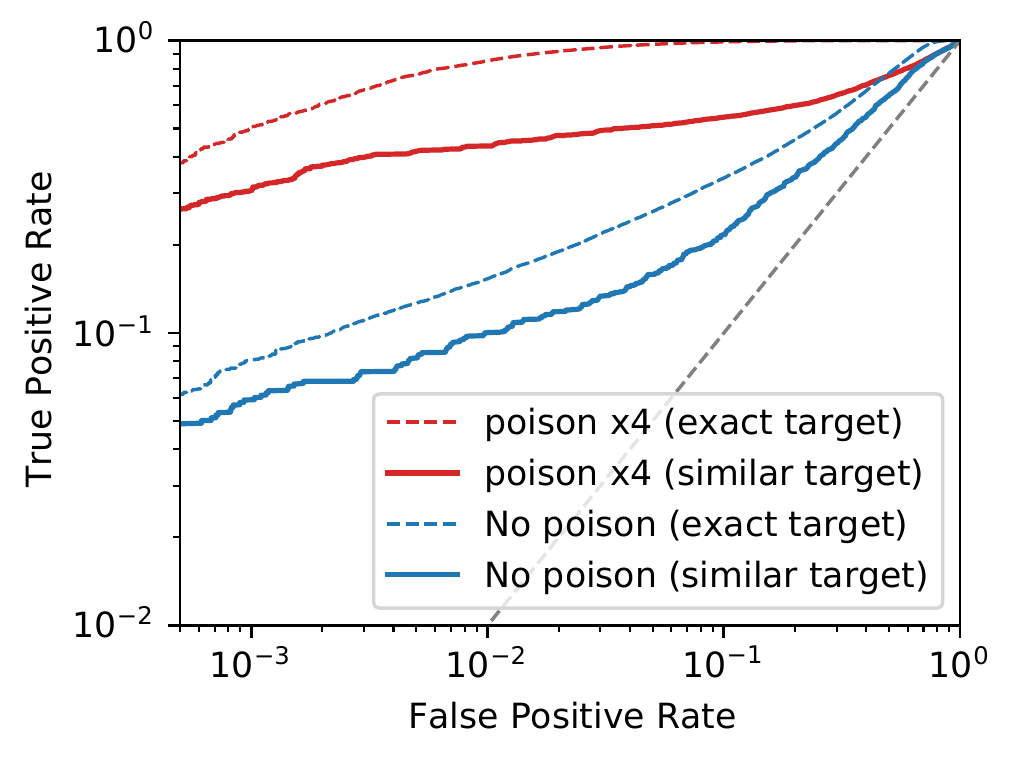}
    \caption{Our MI attack (with 4 poisons) works on CIFAR-10 even when the adversary does not know the exact target, but only a near neighbor.}
    \label{fig:mi_cifar10_influence}
\end{figure}

\subsection{Bounding Outlier Influence with Loss Clipping}
\label{apx:mi_clipping}

In \Cref{fig:mi_cifar10_dists_CE_clipped}, we show the distribution of losses for individual CIFAR-10 examples, for models trained with loss clipping (see \Cref{ssec:mi_ablation_clipping}). Similarly to \Cref{fig:mi_cifar10_dists_CE}, we find that poisoning shifts the model's losses because the poisoned model becomes less confidence in the target example. However, poisoning does not help in making the distributions more separable. On the contrary, as we increase the number of poisons, even examples that were originally easy to infer membership on become hard to distinguish.

\begin{figure}[H]
    \centering
    \includegraphics[width=0.8\columnwidth]{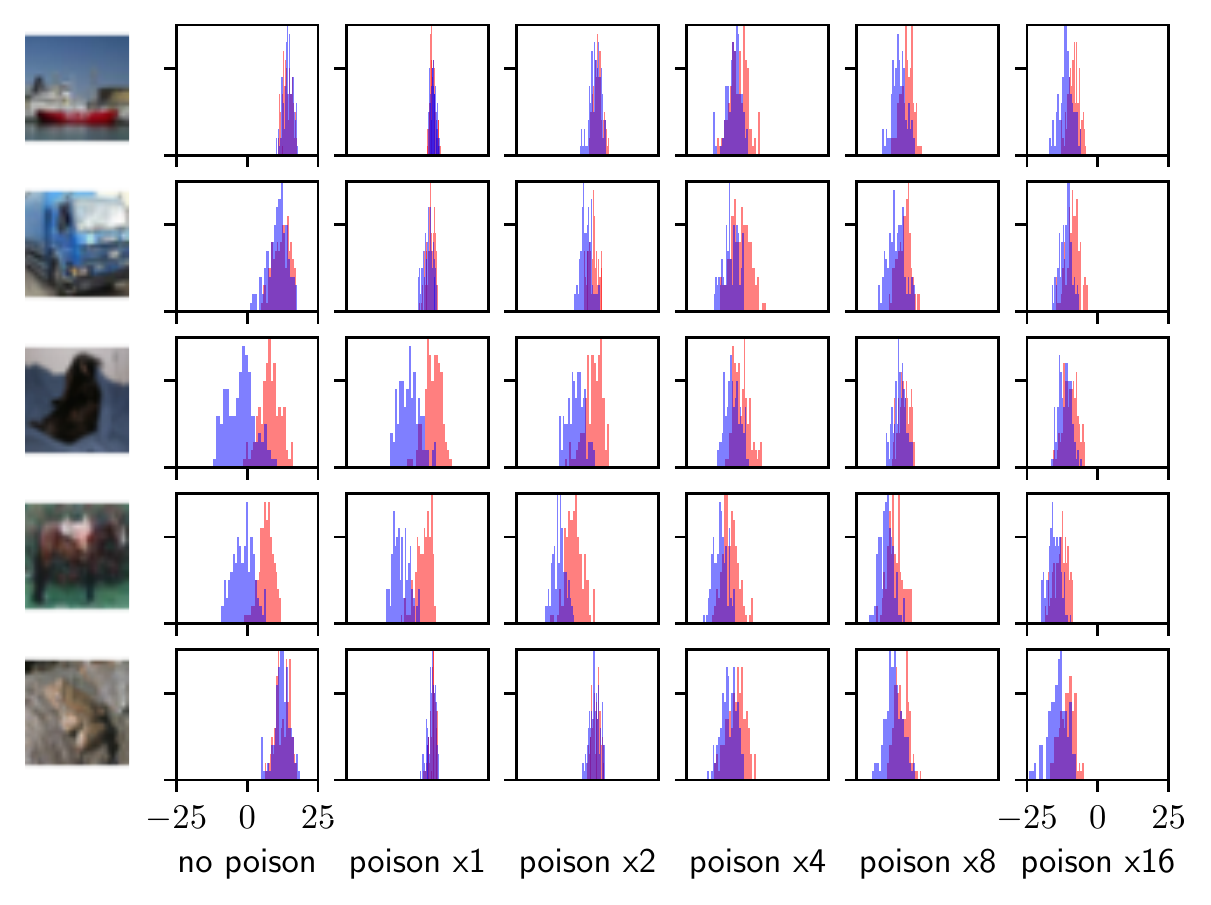}
    \caption{For models trained with clipped losses, poisoning shifts the loss distributions of members (red) and non-members (blue), but does not make them more separable.}
    \label{fig:mi_cifar10_dists_CE_clipped}
\end{figure}


\subsection{Untargeted Membership Inference Attacks}
\label{apx:mi_untargeted}
\paragraph{Alternative strategies and datasets.}


In \Cref{fig:untargeted_cifar10_different_poison}, \Cref{fig:untargeted_cifar100_diff_poisons}, and \Cref{fig:svm_texas100}, we show the results of different untargeted poisoning strategies on CIFAR-10 and CIFAR-100, as well as for an SVM classifier trained on the Texas100 dataset (see~\cite{shokri2016membership} for details on this dataset).

\begin{figure}[H]
    \centering
    \includegraphics[width=0.8\columnwidth]{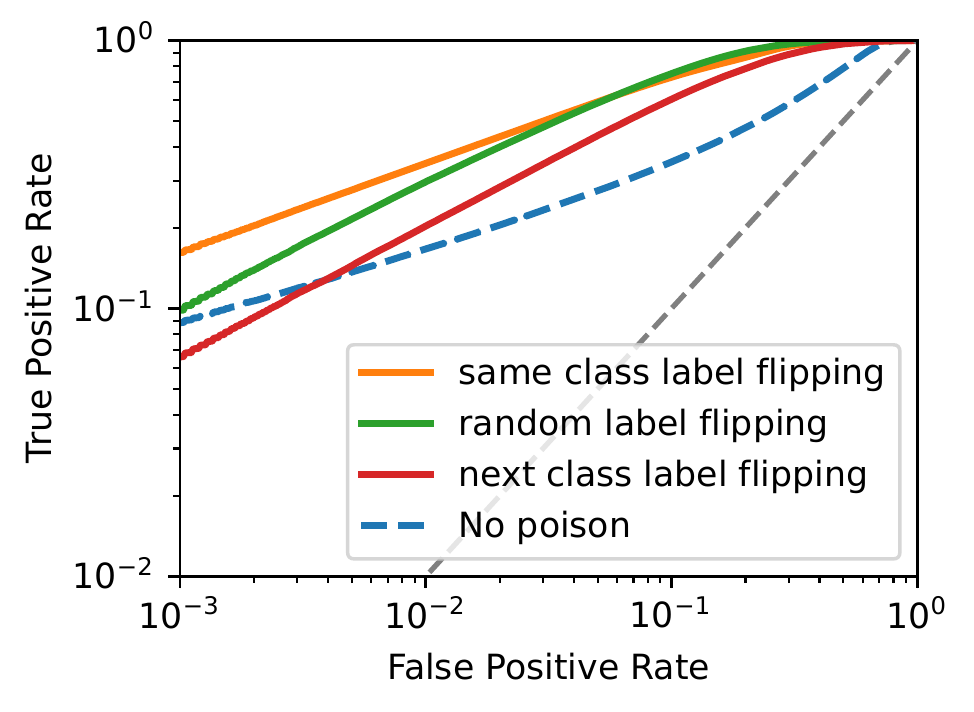}
    \caption{Comparison of untargeted poisoning attacks on CIFAR-10.}
    \label{fig:untargeted_cifar10_different_poison}
\end{figure}
\begin{figure}[H]
    \centering
    \includegraphics[width=0.8\columnwidth]{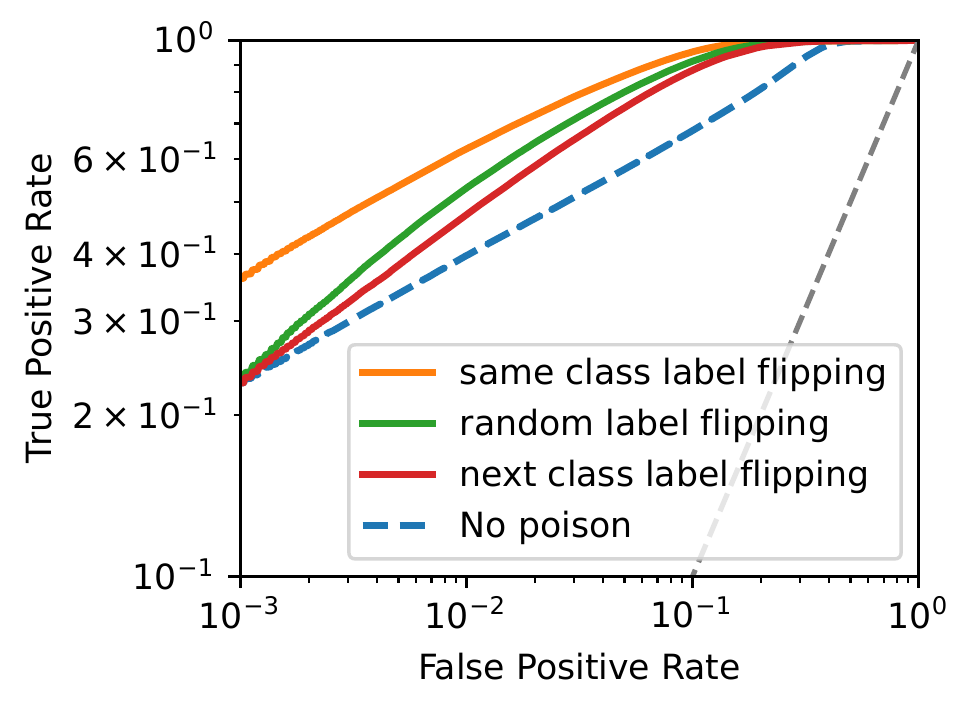}
    \caption{Comparison of untargeted poisoning attacks on CIFAR-100.}
    \label{fig:untargeted_cifar100_diff_poisons}
\end{figure}
\begin{figure}[H]
    \centering
    \includegraphics[width=0.8\columnwidth]{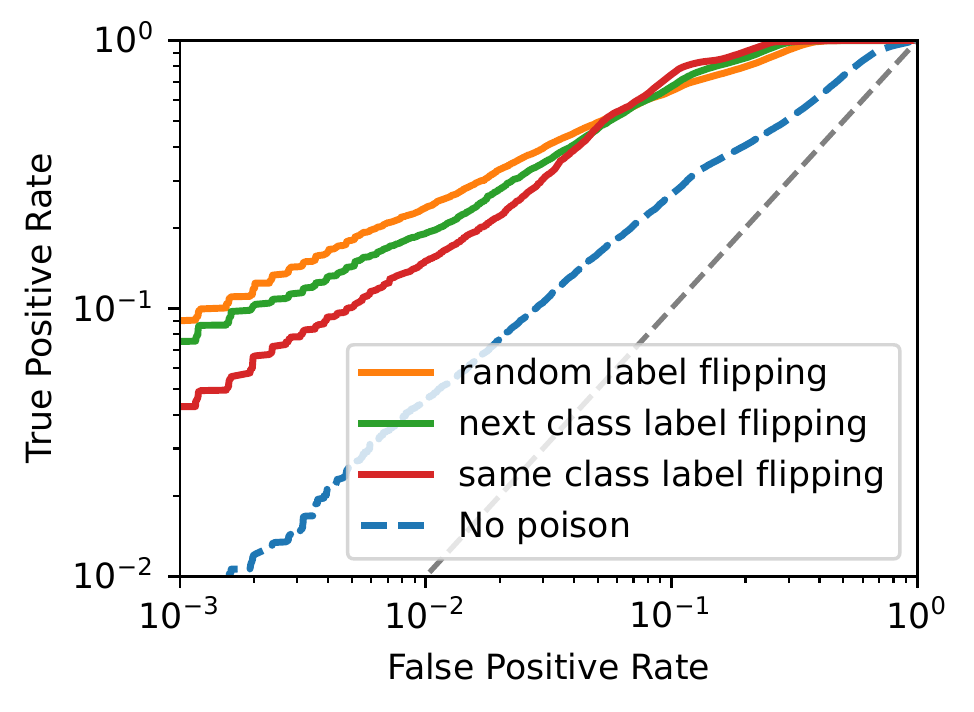}
    \caption{Comparison of untargeted poisoning attacks on Texas100.}
    \label{fig:svm_texas100}
\end{figure}

The best-performing strategy on CIFAR-10 and CIFAR-100, \emph{same class label flipping}, mislabels all of the adversary's points into a single class.
We consider two alternative untargeted poisoning strategies: \emph{random label flipping} where each of the adversary's points is randomly mislabeled into an incorrect class, and \emph{next class label flipping} where the adversary mislabels each example $(x, y)$ into the next class $(x, y+1 \mod |\mathcal{Y}|)$.
On both CIFAR-10 and CIFAR-100, consistently mislabelling all poisoned examples into the same class results in the strongest attack. On the Texas100 dataset, simply mislabeling the adversary's data at random performs slightly better.

On CIFAR-10, we also experimented with strategies where the adversary's share of the data is \emph{out-of-distribution}, e.g., by using randomly mislabeled images from CIFAR-100 or MNIST, or simply images that consist of random noise. However, we could not find a poisoning strategy that performed as well as consistently mislabelling \emph{in-distribution} data.

\paragraph{Distribution of confidences.}
Similarly to the targeted attack, the untargeted poisoning attack also makes the MI attack easier by making individual examples' confidence distributions more separable. In~\Cref{fig:per_example_loss}, we pick five random CIFAR-10 examples and plot the logit-scaled confidence of the data point when it is a member (red) and not a member (blue). In the unpoisoned model (leftmost column), the two distributions overlap for most examples. With an untargeted poisoning attack, the confidences decrease and the distributions become more separable, which makes membership inference easier. 

\begin{figure}[H]
    \centering
    \includegraphics[width=0.7\columnwidth]{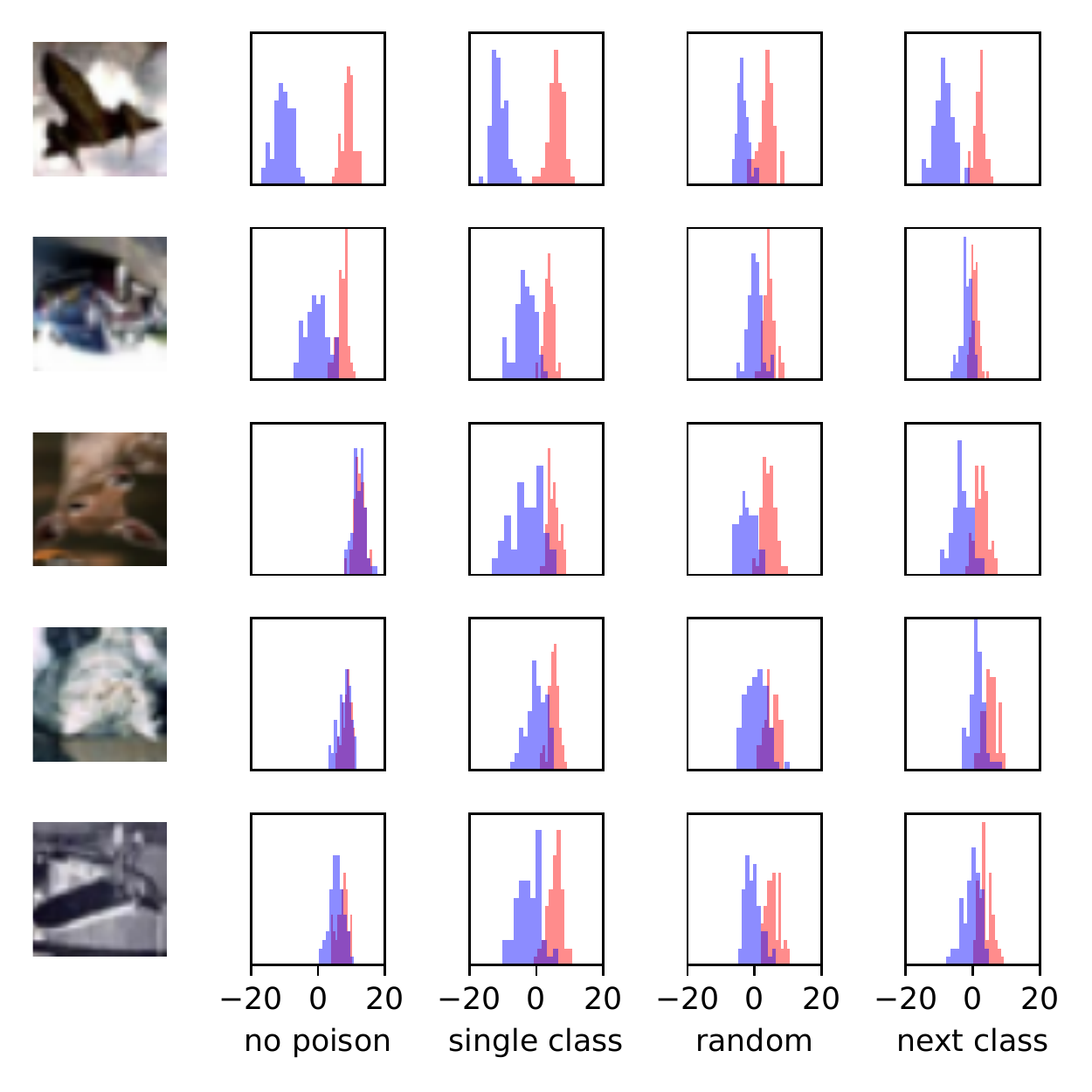}
    \caption{Untargeted poisoning makes the loss distributions of members and non-members easier to distinguish. For five randomly chosen data points, we show the distribution of models' losses (in logit scale) on that example when it is a member (red) and when it is not (blue). The x-axis shows different types of untargeted poisoning strategies.}
    \label{fig:per_example_loss}
\end{figure}

\paragraph{Disparate impact of untargeted poisoning.}

To examine which points are most vulnerable to the untargeted poisoning attack, we perform the same analysis as in \Cref{sssec:extreme_points}. We first pick out the 5\% of least- and most-vulnerable points for a set of models trained without poisoning. We then run an MI attack on both types of points (for a new set of models) with and without poisoning in~\Cref{fig:mi_cifar10_untargeted_extremes}. Untargeted poisoning does not significantly affect the points that were initially most vulnerable. For the points that are hardest to attack without poisoning, our untargeted attack increases the TPR at a 0.1\% FPR from $0.1\%$ to $3.7\%$---an improvement of $37 \times$.
\begin{figure}[H]
    \centering
    \includegraphics[width=0.7\columnwidth]{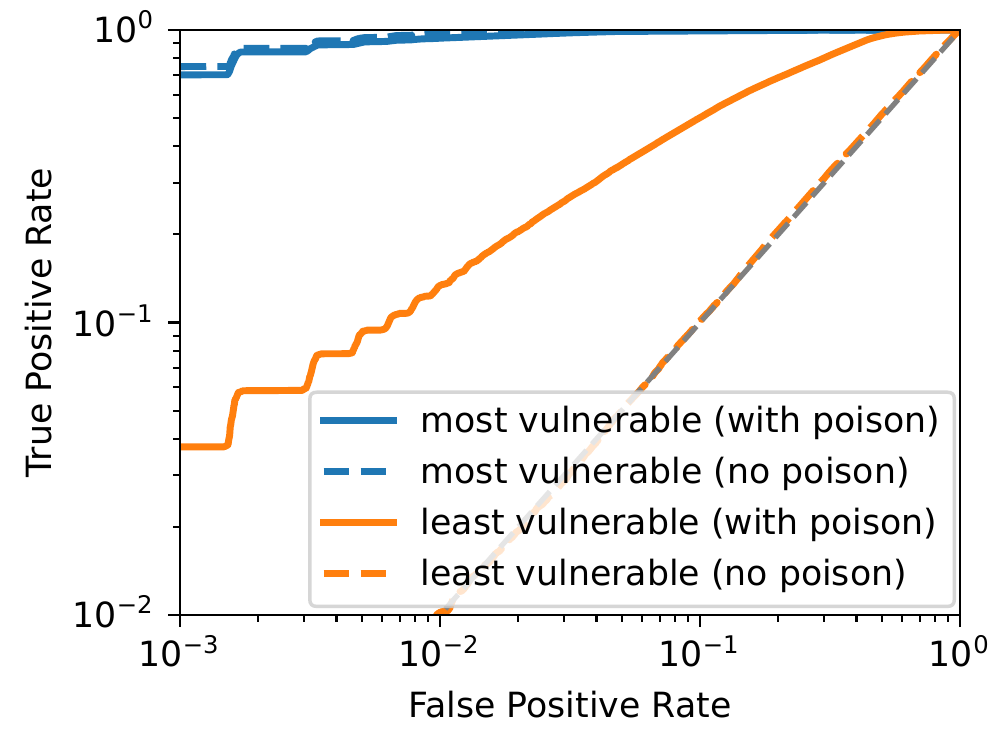}
    \caption{Untargeted poisoning causes previously-safe data points to become vulnerable.
    While poisoning has little effect on the most vulnerable points, poisoning the least vulnerable points improves the TPR at a 0.1\% FPR by $37\times$.}
    \label{fig:mi_cifar10_untargeted_extremes}
\end{figure}
\clearpage
\section{Additional Experiments for Attribute Inference Attacks}

In \Cref{fig:ai_adult_married}, we replicate the experiment from \Cref{ssec:ai_results} but infer a user's relationship status (``married'' or ``non-married'') rather than their gender.
The attack and experimental setup are the same as described in \Cref{ssec:ai_setup}.

The results, shown in \Cref{fig:ai_adult_married} are qualitatively similar as those for inferring gender in \Cref{fig:ai_adult_gender}.
At a FPR of 0.1\%, the attack of~\cite{mehnaz2022your} (without poisoning) achieves a TPR of 4\%, while our attack with 16 poisons obtains a TPR of 18\%. 
At false-positive rates of \textgreater$5\%$ all the attribute inference attacks (even with poisoning) perform worse than a trivial imputation baseline.

\begin{figure}[H]
    \centering
    \includegraphics[width=0.8\columnwidth]{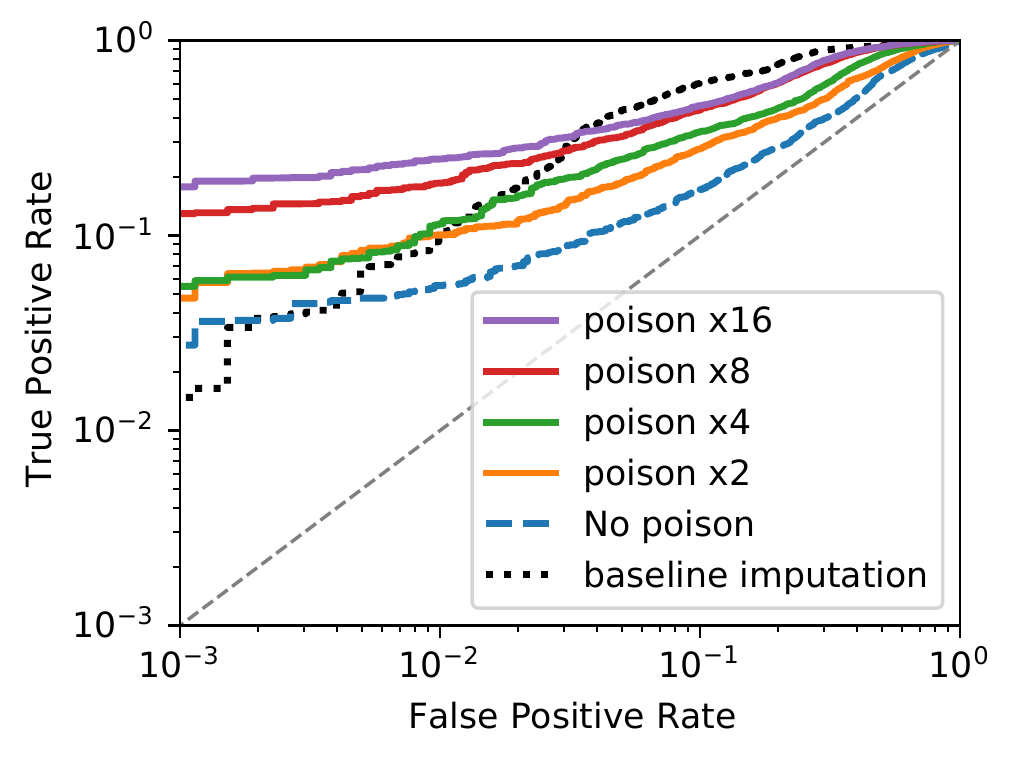}
    \caption{Targeted poisoning attacks boost attribute inference (for inferring relationship status) on Adult. Without poisoning, the attack of~\cite{mehnaz2022your} performs no better than a baseline imputation that infers relationship status based on correlations with other attributes. With poisoning, the attack significantly outperforms the baseline at low false-positives.}
    \label{fig:ai_adult_married}
\end{figure}

\section{Additional Experiments for Canary Extraction}
\label{apx:lm}

\subsection{Strategies for Suffix Poisoning}
\label{apx:lm_suffix}

In \Cref{ssec:lm_suffix} we showed that we could increase exposure after poisoning a canary's prefix by re-inserting it multiple times into the training set padded with zeros.
In \Cref{fig:lm_poison_all}, we consider alternative suffix poisoning strategies, that are ultimately less effective.
Padding the prefix with a list of random tokens or a random 6-digit number also provides a moderate increase in exposure (to $8.5$ bits and $8.3$ bits respectively), as long as the same random suffix is re-used for all poisons. If we insert the poison many times with \emph{different} random suffixes, the poisoning actually \emph{hurts} the attack. This mirrors our finding  in \Cref{fig:mi_cifar10_target} and \Cref{fig:mi_cifar100_target} that mislabeling a target point with different incorrect labels hurts MI attacks.

\begin{figure}[H]
    \centering
    \includegraphics[width=0.8\columnwidth]{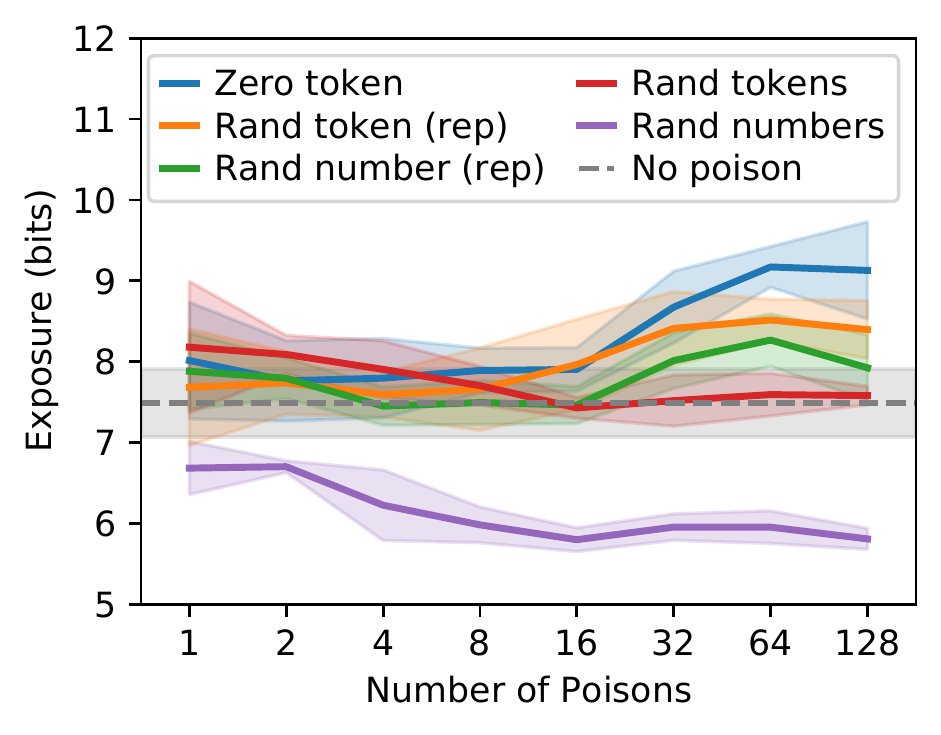}
    \caption{Poisoning a canary's prefix by padding it with zeros is more effective than alternative strategies that pad with random tokens or random 6-digit numbers. Replicating the same padding for each poisoned copy (rep) is much more effective than using a different random padding for each copy.}
    \label{fig:lm_poison_all}
\end{figure}

\subsection{Canary Extraction on a Fixed Budget}
\label{apx:lm_guesses}

In \Cref{sec:lm}, we evaluated canary extraction attacks in terms of the \emph{average} exposure of different canaries inserted into a training set. An increase in average exposure does not necessarily tell us whether the attack is making extraction of canaries more practical (e.g., an attack might allow the adversary to recover canaries that used to require $200{,}000$ guesses in ``only'' $100{,}000$ guesses, without making any difference for those canaries that can be extracted in less than $100{,}000$ guesses).
This is not the case for our attack. As we show in \Cref{fig:lm_top100}, poisoning increases the attacker's success rate in extracting canaries for any budget of guesses.
For example, if the adversary is limited to 100 guesses for a canary, their success rate grows from $10\%$ without poisoning to $41\%$ with poisoning.

\begin{figure}[H]
    \centering
    \includegraphics[width=0.8\columnwidth]{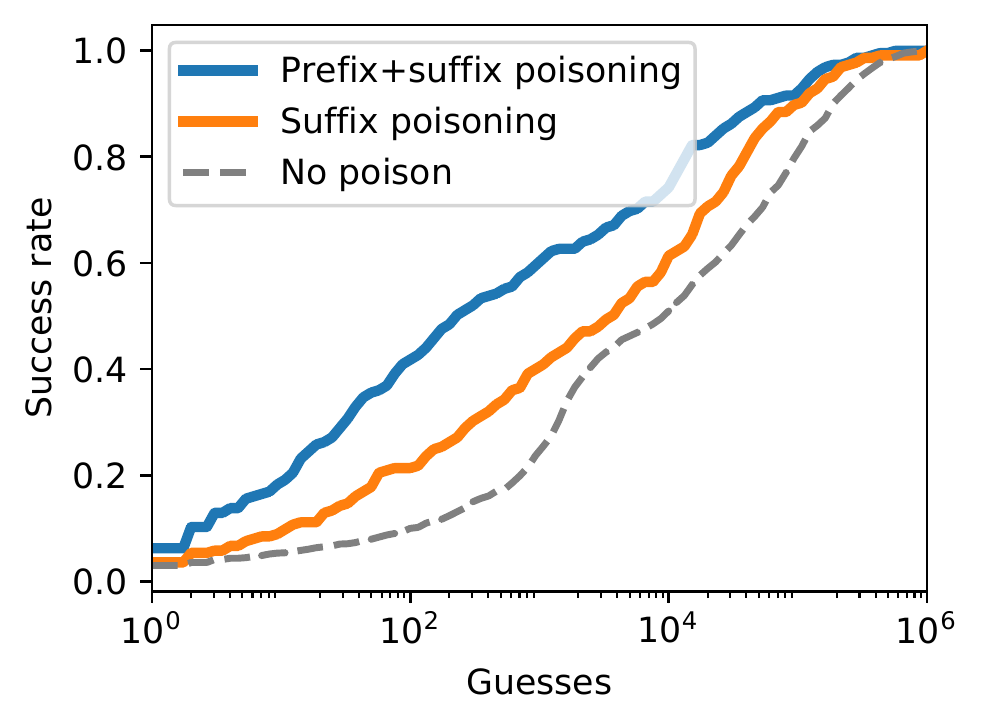}
    \caption{For any fixed budget of guesses, poisoning increases the attacker's success rate in recovering a secret canary.}
    \label{fig:lm_top100}
\end{figure}

\clearpage
\section{Provably Amplifying Privacy Leakage}
\label{apx:theory}

In this section, we provide additional theoretical analysis that proves a targeted poisoning attack can achieve perfect membership inference in the case of k-Nearest Neighbors (kNNs) and linear Support Vector Machines (SVMs).

\paragraph{Poisoning $k$-nearest neighbor classifiers.}

\begin{algorithm}
\SetAlgoLined
\DontPrintSemicolon
\KwData{Target point $(x, y)$, nearest neighbor count $k$, minimum distance $\delta$}
\SetKwFunction{kNNPoison}{\textsc{kNNPoison}}
\SetKwFunction{MI}{\textsc{MI}}

\SetKwProg{Fn}{Function}{:}{}
\Fn{\kNNPoison{$x, y, k, \delta$}}{
    Pick an incorrect label $y' \neq y$\;
    $X_{\text{adv}}=\{\underbrace{x, \dots, x}_{k-1}\}$ \tcp*{Make $k-1$ copies of x}
    $Y_{\text{adv}}=\{\underbrace{y, \dots, y}_{(k-1)/2}, \underbrace{y', \dots, y'}_{(k-1)/2}\}$ \tcp*{Evenly balance classes}
    $X_{\text{adv}} = X_{\text{adv}} \cup \{x'\}, \text{ such that } \|x - x'\| = \delta$\;
    $Y_{\text{adv}} = Y_{\text{adv}} \cup \{y'\}$ \tcp*{Mislabel next-closest sample}
    \KwRet $D_{\text{adv}} = X_{\text{adv}}, Y_{\text{adv}}$\;
}
\;
\Fn{\MI{$x, y, f_{\text{kNN}}$}}{
    $\hat{y} \gets f_{\text{kNN}}(x)$  \tcp*{Query the model (as a black-box)}
    \eIf{$\hat{y} = y$}{\KwRet ``member''}{\KwRet ``non-member''}
}
\caption{$k$-nearest neighbors poisoning}
\label{alg:knnpois}
\end{algorithm}

In \Cref{ssec:attack} we introduced a strategy that used poisoning to obtain 100\% membership inference accuracy on a targeted point for kNNs (see Algorithm~\ref{alg:knnpois}). Here, we show that poisoning is indeed \emph{necessary} to obtain such a strong attack. To make this argument, we prove that without poisoning there exist points where membership inference cannot succeed better than chance. 

We say that a point $(x,y) \in D$ is \emph{unused} by the model if the model's output on \emph{any} point is unaffected by the removal of $(x,y)$ from the training set $D$. Such points are easy to construct: e.g., consider a cluster of $\gg k$ close-by points that all the share the same label. Removing one point from the center of this cluster will not affect the model's output on any input (a simple one-dimensional example visualization is given in \Cref{fig:knn_unused}). 
For any such unused point, inferring membership is impossible: the model's input-output behavior is identical whether the model is trained on $D$ or on $D \setminus \{(x, y)\}$. However, the poisoning strategy in Algorithm~\ref{alg:knnpois} still succeeds on these points, and thus provably increases privacy leakage.

\begin{figure}[H]
    \centering
    \includegraphics[width=\linewidth]{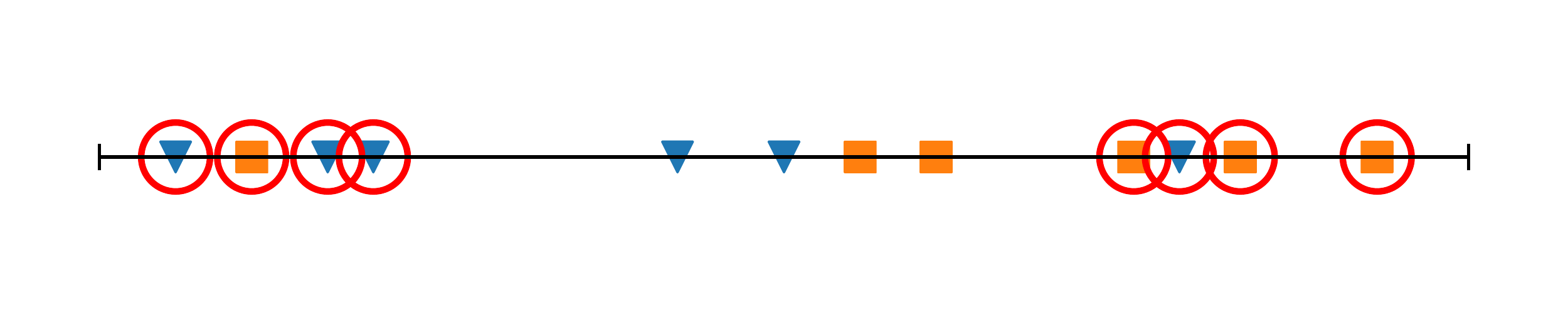}
    \vspace{-3.5em}
    \caption{Unused points in a one-dimensional $k$-nearest neighbors classifier with $k=3$. Blue triangles are from class $0$, and orange squares from class $1$. Removing one of the circle ``unused'' training samples does not affect the model's decision on any test point.}
    \label{fig:knn_unused}
\end{figure}

\paragraph{Poisoning support vector machines.}

Here, we consider an adversary who receives black-box access to a linear SVM.
By definition, only support vectors are used at inference time and thus distinguishing between an SVM trained on $D$ and one trained on $D / \lbrace (x, y) \rbrace$ is impossible unless the sample $(x, y)$ is a support vector in at least one of these two models. 
We show that there exist points that can be forced---by a poisoning attack---to become support vectors if they are members of the training set.
By computing the distance between the poisoned points and the classifier's decision boundary (which can be done with black-box model access), the adversary can then infer with 100\% accuracy whether some targeted point was a member or not.

Unlike for $k$-nearest neighbors models, we will not be able to reveal membership for \emph{any} point. Instead, our attack can only succeed on examples that lie on the convex hull of examples from one class. (However note that in high dimensions almost all points are on the boundary of the convex hull, and almost no points are contained in the interior.)
We propose a sufficient condition for such a point to be forced to be a support vector, which we call \emph{protruding}:
\begin{definition}
For a binary classification dataset $D$, a point $(x_t, y_t)\in D$ is \emph{protruding} if there exists some $w, b_0, b_1$ so that the plane $w\cdot x + b_0=0$ linearly separates $D$ and $w\cdot x + b_1=0$ linearly separates $D / \lbrace (x_t, y_t)\rbrace \cup \lbrace (x_t, 1-y_t)\rbrace$.
\end{definition}

\begin{figure}[t]
    \centering
    \includegraphics[width=\linewidth]{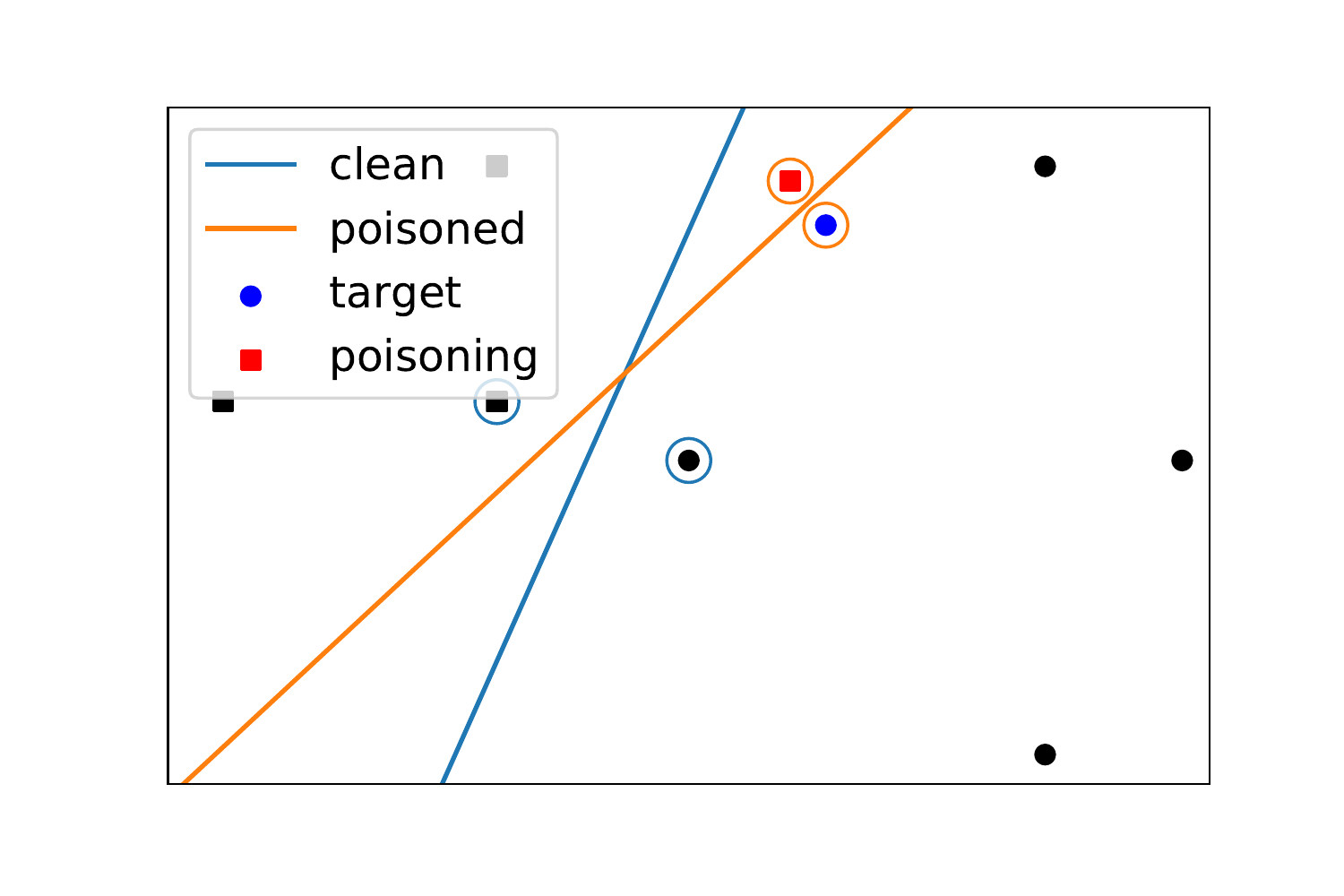}
    \vspace{-2em}
    \caption{Consider a squares versus circles classification task.
    By inserting the red square as a poison, we can now perform membership inference on the target blue circle. But because this blue circle was not a support vector before poisoning, membership inference was impossible. 
    By adding the red square poisoned point, the maximum-margin classifier shifts to the orange line, and the target becomes a support vector. (The support vectors for each model are circled in the model's color.)}
    \label{fig:svm_theory}
\end{figure}

Intuitively, this definition says that a point is protruding if there exists some way to linearly separate the two classes, such that this protruding point is the closest training example to the decision boundary. Then, if that point's label were flipped, it suffices to ``shift'' the decision boundary (i.e., by modifying the offset $b$) to linearly separate the data again.
We give an example of a protruding point (the target point) in Figure~\ref{fig:svm_theory}. If a point is protruding, we can insert a poisoned point of the opposite class close to it to force the protruding point to become a support vector.

\begin{theorem}
Let $D$ be a binary classification dataset containing a protruding point $(x_t, y_t)$. Then there exists some $(x_p, y_p)$ so that $D\cup \lbrace (x_p, y_p)\rbrace$ has $(x_t, y_t)$ as a support vector, and a larger margin when $(x_t, y_t)\notin D$. 
\end{theorem}
\begin{proof}
Without loss of generality, assume $y_t=0$. Because $(x_t, y_t)$ is protruding, we know there exists some $w, b_0, b_1$ satisfying the conditions of the definition. Write $b$ such that $f(x)=w\cdot x + b$ has $f(x_t)=0$. Let $\delta$ be the distance from the plane $f(x)=0$ to the nearest point in $D / \lbrace (x_t, y_t)\rbrace$. We have $\delta>0$ because the plane $f$ lies strictly in between the planes $w\cdot x + b_0=0$ and $w\cdot x + b_1=0$, which both linearly separate $D / \lbrace (x_t, y_t)\rbrace$.

Then consider the poisoning $(x_p, y_p) = (x_t + w\tfrac{\delta}{2||w||}, 1)$. When $(x_t, y_t)\in D$, the maximum margin separator of $D\cup \lbrace (x_p, y_p)\rbrace$ is $w\cdot x + b+\tfrac{\delta}{4||w||}=0$. The distance from each point in $D / \lbrace (x_t, y_t)\rbrace$ to this plane must be at least $\tfrac{3\delta}{4}$, as this has shifted $f$ by a distance of $\delta/4$. Then $(x_t, y_t)$ will be a support vector of this plane, with a margin of $\delta/4$.

When $(x_t, y_t)\notin D$, the margin of the resulting hyperplane must be larger than $\delta/4$, as $f$ is a hyperplane which linearly separates $D\cup \lbrace (x_p, y_p)\rbrace$ with a margin of $\delta/2$.
\end{proof}

Our analysis here assumes that the adversary knows everything about the training set except for whether $(x_t, y_t)$ is a member, and that the dataset is linearly separable.

We also run a brief experiment to show that \emph{untargeted} white-box attacks on SVMs are also possible. Given white-box access to an SVM, the adversary can directly recover the data of the support vectors, as these are necessary to perform inference. Our untargeted attacks increase privacy leakage by forcing the trained model to use \emph{more} data points as support vectors. We train SVMs on Fashion MNIST restricted to the first two classes, using 2000 points for training and injecting 200 poisoning points according to a simple label flipping strategy. Over 5 trials, an unpoisoned linear SVM has an average of 121 support vectors, and an unpoisoned polynomial kernel SVM has an average of 176 support vectors. When adding the label flipping attack, the poisoned linear SVM grows to 512 support vectors from the clean training set, and the polynomial SVM grows to 642 support vectors from the clean training set, increasing the number of leaked data points by a factor of $4.2\times$ and $3.7\times$, respectively.

\end{document}